\documentclass{article}
\usepackage{color}
\usepackage[all]{xy}
\CompileMatrices
\usepackage{qsymbols}
\usepackage{amsmath,mathrsfs} 
\usepackage{QED}
\usepackage{fancybox} 
\usepackage[dvips]{graphicx}
\usepackage{pstricks,pst-node} \psset{rowsep=.7cm,colsep=.7cm}
\psset{arrowsize=3.5pt 3.8,arrowlength=0.7}
\psset{labelsep=.1cm,nodesep=2pt}

 \newif\ifcomment
 \commentfalse 

\definecolor{darkbrown}{cmyk}{.3,.75,.75,.15}
\newcommand{\brown}[1]{{\color{darkbrown} #1}}
\newcommand{\rouge}[1]{{\color{red} #1}}
\newcommand{\bl}[1]{\textcolor{blue}{#1}}
\newcommand{\verte}[1]{{\color{green}#1}}
\definecolor{vertfonce}{rgb}{0,.5,0}
\newcommand{\verdir}[1]{{\color{vertfonce} #1}}

\renewcommand {\l}{\langle} \renewcommand {\r}{\rangle}

\newtheorem{definition}{Definition}
\newtheorem{lemma}{Lemma}
\newtheorem{theorem}{Theorem}
\newtheorem{example}{Example}

\newcommand{\Alice}{\textit{Alice}}
\newcommand{\Beth}{\textit{Beth}}

\newcommand{\conv}[1]{"-->"_{#1}}
\newcommand{\pref}[1]{".>"_{#1}}
\newcommand{\prefAlice}{\raisebox{2pt}{\hspace{-30pt}
  \begin{psmatrix}
    &[name=a]&[name=b]\ncline[arrows=->,linewidth=.04,linecolor=red,linestyle=dotted,arrowscale=.5]{a}{b}
  \end{psmatrix}
}}
\newcommand{\prefBeth}{\raisebox{2pt}{\hspace{-23pt}
  \begin{psmatrix}
    &[name=a]&[name=b]\ncline[arrows=->,linewidth=.015,linecolor=blue,linestyle=dotted,arrowscale=.5]{a}{b}
  \end{psmatrix}
}}
\newcommand{\ConvA}{\rouge{\dashrightarrow}}
\newcommand{\ConvB}{\bl{\dashrightarrow}}
\newcommand{\CoM}[1]{\xymatrix @C 15pt{\ar@{->}[r]&}_{#1}}
\newcommand{\gCoM}{\xymatrix @C 20pt{\ar@{->}[r]&}}
\newcommand{\bydef}{\triangleq}
\newcommand{\eqdef}{\mathop{\triangleq}}

\newcommand{\A}{\ensuremath{\mathcal{A}}}
\renewcommand{\S}{\ensuremath{\mathcal{S}}}

\newcommand{\aConvFct}[2]{\ensuremath{\mathrel{\xymatrix@C3pt{\ar@{>-}^{#1}_{#2}[r]&}}}\!\!}

\newcommand{\aConvEqFct}[2]{\ensuremath{\mathrel{\xymatrix@C3pt{\ar@{>-<}^{#1}_{#2}[r]&}}}}

\newcommand{\lAgentConvFct}[2]{\ensuremath{\mathop{\xymatrix@C3pt{\ar@{>-}^{#1}_{#2}[r]&}}}}

\newcommand{\EqFctName}[2] {\mathrm{Eq}^{#1}_{#2}}
\newcommand{\EqFct}[3] {\EqFctName{{#1}}{{#2}}({#3})}

\DeclareMathSymbol{\leggedrightarrow}{\mathord}{AMSa}{"4B}

\newcommand{\freeComName}[1] {\comName{#1}{}}

\newcommand{\freeComsName}[1] {{\rightarrow_{#1}^*}}
\newcommand{\freeComs}[3]{{#2}\freeComsName{#1}{#3}}
\newcommand{\comName}[2] {\rightarrow_{#1}^{#2}}

\newcommand{\Syn}{\mathcal{S}} \newcommand{\syns}{\mathrm{C}}
\newcommand{\syn}{\mathrm{s}} 

  \newcommand{\optionsName} {\mathscr{U}}
\newcommand{\options}[2] {\optionsName({#1})}

\newcommand{\Gname} {\mathrm{G}} 
 
\newcommand{\cpG}[1]{\ensuremath{\Gname_{{#1}}}}

\newcommand{\agamewithtwoSCCs}{
  \begin{figure}
    \centering
    \[
\begin{psmatrix}[colsep=.4cm,rowsep=.4cm]
  &&&&& [name=42]{\mathbf{4|2}}\\
  &&&& [name=43]{\mathbf{4|3}} && [name=41]{\mathbf{4|1}}\\
  &&& [name=12]{\mathbf{1|2}} &&&& [name=32]{\mathbf{3|2}} \\
  &[name=1w]{\mathbf{\rouge{1|\omega}}}& [name=13]{\mathbf{1|3}} &&&&&& [name=31]{\mathbf{3|1}} \\
  &&& [name=14]{\mathbf{1|4}} &&&& [name=34]{\mathbf{3|4}} \\
  &&&& [name=23]{\mathbf{2|3}} && [name=21]{\mathbf{2|1}} \\
  &&&&& [name=24]{\mathbf{2|4}}
 \ncarc[arrows=->,linewidth=.03,arcangle=50]{1w}{12}
 \ncarc[arrows=->,linewidth=.03]{1w}{13}
 \ncarc[arrows=->,linewidth=.03,arcangle=-50]{1w}{14}
 \ncarc[arrows=->,linewidth=.03,arcangle=30]{41}{42}
 \ncarc[arrows=->,linewidth=.03,arcangle=30]{42}{43}
 \ncarc[arrows=->,linewidth=.03,arcangle=-30]{12}{13}
 \ncarc[arrows=->,linewidth=.03,arcangle=-30]{13}{14}
 \ncarc[arrows=->,linewidth=.03,arcangle=-30]{31}{32}
 \ncarc[arrows=->,linewidth=.03,arcangle=-30]{34}{31}
 \ncarc[arrows=->,linewidth=.03,arcangle=30]{23}{24}
 \ncarc[arrows=->,linewidth=.03,arcangle=30]{24}{21}
 \ncarc[arrows=->,linewidth=.03,arcangle=-40]{42}{12}
 \ncarc[arrows=->,linewidth=.03,arcangle=-40]{32}{42}
 \ncarc[arrows=->,linewidth=.03,arcangle=40]{43}{13}
 \ncarc[arrows=->,linewidth=.03,arcangle=40]{31}{41}
 \ncarc[arrows=->,linewidth=.03,arcangle=40]{13}{23}
 \ncarc[arrows=->,linewidth=.03,arcangle=40]{21}{31}
 \ncarc[arrows=->,linewidth=.03,arcangle=-40]{14}{24}
 \ncarc[arrows=->,linewidth=.03,arcangle=-40]{24}{34}
\end{psmatrix}
 \]
    \caption{A game with two SCC's}
    \label{fig:2SCC}
  \end{figure} 
}

\newcommand{\gametwo}{
  \begin{figure}[t]
    \begin{footnotesize}
      \[
      \begin{psmatrix}[colsep=.4cm,rowsep=.4cm]
        &&&& [name=42]{\mathbf{4|2}}\\
        &&& [name=43]{\mathbf{4|3}} && [name=41]{\mathbf{4|1}}\\
        && [name=12]{\mathbf{1|2}} &&&& [name=32]{\mathbf{3|2}} \\
        & [name=13]{\mathbf{1|3}} &&&&&& [name=31]{\mathbf{3|1}} \\
        && [name=14]{\mathbf{1|4}} &&&& [name=34]{\mathbf{3|4}} \\
        &&& [name=23]{\mathbf{2|3}} && [name=21]{\mathbf{2|1}} \\
        &&&& [name=24]{\mathbf{2|4}}
        \ncarc[arrows=->,linewidth=.03,linestyle=dashed,linecolor=blue,arcangle=30]{41}{42}
        \ncarc[arrows=->,linewidth=.03,linestyle=dashed,linecolor=blue,arcangle=30]{42}{43}
        \ncarc[arrows=->,linewidth=.03,linestyle=dashed,linecolor=blue,arcangle=-30]{12}{13}
        \ncarc[arrows=->,linewidth=.03,linestyle=dashed,linecolor=blue,arcangle=-30]{13}{14}
        \ncarc[arrows=->,linewidth=.03,linestyle=dashed,linecolor=blue,arcangle=-30]{31}{32}
        \ncarc[arrows=->,linewidth=.03,linestyle=dashed,linecolor=blue,arcangle=-30]{34}{31}
        \ncarc[arrows=->,linewidth=.03,linestyle=dashed,linecolor=blue,arcangle=30]{23}{24}
        \ncarc[arrows=->,linewidth=.03,linestyle=dashed,linecolor=blue,arcangle=30]{24}{21}
        \ncarc[arrows=->,linewidth=.03,linestyle=dashed,linecolor=red,arcangle=-40]{42}{12}
        \ncarc[arrows=->,linewidth=.03,linestyle=dashed,linecolor=red,arcangle=-40]{32}{42}
        \ncarc[arrows=->,linewidth=.03,linestyle=dashed,linecolor=red,arcangle=40]{43}{13}
        \ncarc[arrows=->,linewidth=.03,linestyle=dashed,linecolor=red,arcangle=40]{31}{41}
        \ncarc[arrows=->,linewidth=.03,linestyle=dashed,linecolor=red,arcangle=40]{13}{23}
        \ncarc[arrows=->,linewidth=.03,linestyle=dashed,linecolor=red,arcangle=40]{21}{31}
        \ncarc[arrows=->,linewidth=.03,linestyle=dashed,linecolor=red,arcangle=-40]{14}{24}
        \ncarc[arrows=->,linewidth=.03,linestyle=dashed,linecolor=red,arcangle=-40]{24}{34}
      \end{psmatrix}
      \quad
      \begin{psmatrix}[colsep=.4cm,rowsep=.4cm]
        &&&& [name=42]{\mathbf{4|2}}\\
        &&& [name=43]{\mathbf{4|3}} && [name=41]{\mathbf{4|1}}\\
        && [name=12]{\mathbf{1|2}} &&&& [name=32]{\mathbf{3|2}} \\
        & [name=13]{\mathbf{1|3}} &&&&&& [name=31]{\mathbf{3|1}} \\
        && [name=14]{\mathbf{1|4}} &&&& [name=34]{\mathbf{3|4}} \\
        &&& [name=23]{\mathbf{2|3}} && [name=21]{\mathbf{2|1}} \\
        &&&& [name=24]{\mathbf{2|4}}
        \ncarc[arrows=->,linewidth=.03,linestyle=dotted,arcangle=30]{41}{42}
        \ncarc[arrows=->,linewidth=.03,linestyle=dotted,arcangle=30]{41}{43}
        \ncarc[arrows=->,linewidth=.03,linestyle=dotted,arcangle=30]{42}{43}
        \ncarc[arrows=->,linewidth=.03,linestyle=dotted,arcangle=-30]{12}{13}
        \ncarc[arrows=->,linewidth=.03,linestyle=dotted,arcangle=30]{12}{14}
        \ncarc[arrows=->,linewidth=.03,linestyle=dotted,arcangle=-30]{13}{14}
        \ncarc[arrows=->,linewidth=.03,linestyle=dotted,arcangle=-30]{31}{32}
        \ncarc[arrows=->,linewidth=.03,linestyle=dotted,arcangle=-30]{34}{31}
        \ncarc[arrows=->,linewidth=.03,linestyle=dotted,arcangle=30]{34}{32}
        \ncarc[arrows=->,linewidth=.03,linestyle=dotted,arcangle=30]{23}{24}
        \ncarc[arrows=->,linewidth=.03,linestyle=dotted,arcangle=30]{23}{21}
        \ncarc[arrows=->,linewidth=.03,linestyle=dotted,arcangle=30]{24}{21}
        \ncarc[arrows=->,linewidth=.03,linestyle=dotted,arcangle=-40]{42}{12}
        \ncarc[arrows=->,linewidth=.03,linestyle=dotted,arcangle=-40]{32}{42}
        \ncarc[arrows=->,linewidth=.03,linestyle=dotted,arcangle=30]{32}{12}
        \ncarc[arrows=->,linewidth=.03,linestyle=dotted,arcangle=40]{43}{13}
        \ncarc[arrows=->,linewidth=.03,linestyle=dotted,arcangle=40]{43}{23}
        \ncarc[arrows=->,linewidth=.03,linestyle=dotted,arcangle=40]{31}{41}
        \ncarc[arrows=->,linewidth=.03,linestyle=dotted,arcangle=40]{13}{23}
        \ncarc[arrows=->,linewidth=.03,linestyle=dotted,arcangle=40]{21}{31}
        \ncarc[arrows=->,linewidth=.03,linestyle=dotted,arcangle=40]{21}{41}
        \ncarc[arrows=->,linewidth=.03,linestyle=dotted,arcangle=-40]{14}{24}
        \ncarc[arrows=->,linewidth=.03,linestyle=dotted,arcangle=40]{14}{34}
        \ncarc[arrows=->,linewidth=.03,linestyle=dotted,arcangle=-40]{24}{34}
      \end{psmatrix}
      \]
      \[
      \begin{psmatrix}[colsep=.4cm,rowsep=.4cm]
        &&&& [name=42]{\mathbf{4|2}}\\
        &&& [name=43]{\mathbf{4|3}} && [name=41]{\mathbf{4|1}}\\
        && [name=12]{\mathbf{1|2}} &&&& [name=32]{\mathbf{3|2}} \\
        & [name=13]{\mathbf{1|3}} &&&&&& [name=31]{\mathbf{3|1}} \\
        && [name=14]{\mathbf{1|4}} &&&& [name=34]{\mathbf{3|4}} \\
        &&& [name=23]{\mathbf{2|3}} && [name=21]{\mathbf{2|1}} \\
        &&&& [name=24]{\mathbf{2|4}}
        \ncarc[arrows=->,linewidth=.03,linecolor=blue,arcangle=30]{41}{42}
        \ncarc[arrows=->,linewidth=.03,linecolor=blue,arcangle=30]{42}{43}
        \ncarc[arrows=->,linewidth=.03,linecolor=blue,arcangle=-30]{12}{13}
        \ncarc[arrows=->,linewidth=.03,linecolor=blue,arcangle=-30]{13}{14}
        \ncarc[arrows=->,linewidth=.03,linecolor=blue,arcangle=-30]{31}{32}
        \ncarc[arrows=->,linewidth=.03,linecolor=blue,arcangle=-30]{34}{31}
        \ncarc[arrows=->,linewidth=.03,linecolor=blue,arcangle=30]{23}{24}
        \ncarc[arrows=->,linewidth=.03,linecolor=blue,arcangle=30]{24}{21}
        \ncarc[arrows=->,linewidth=.03,linecolor=red,arcangle=-40]{42}{12}
        \ncarc[arrows=->,linewidth=.03,linecolor=red,arcangle=-40]{32}{42}
        \ncarc[arrows=->,linewidth=.03,linecolor=red,arcangle=40]{43}{13}
        \ncarc[arrows=->,linewidth=.03,linecolor=red,arcangle=40]{31}{41}
        \ncarc[arrows=->,linewidth=.03,linecolor=red,arcangle=40]{13}{23}
        \ncarc[arrows=->,linewidth=.03,linecolor=red,arcangle=40]{21}{31}
        \ncarc[arrows=->,linewidth=.03,linecolor=red,arcangle=-40]{14}{24}
        \ncarc[arrows=->,linewidth=.03,linecolor=red,arcangle=-40]{24}{34}
      \end{psmatrix}
      \quad
      \begin{psmatrix}[colsep=.4cm,rowsep=.4cm]
        &&&& [name=42]{\mathbf{4|2}}\\
        &&& [name=43]{\mathbf{4|3}} && [name=41]{\mathbf{4|1}}\\
        && [name=12]{\mathbf{1|2}} &&&& [name=32]{\mathbf{3|2}} \\
        & [name=13]{\mathbf{1|3}} &&&&&& [name=31]{\mathbf{3|1}} \\
        && [name=14]{\mathbf{1|4}} &&&& [name=34]{\mathbf{3|4}} \\
        &&& [name=23]{\mathbf{2|3}} && [name=21]{\mathbf{2|1}} \\
        &&&& [name=24]{\mathbf{2|4}}
        \ncarc[arrows=->,linewidth=.03,arcangle=30]{41}{42}
        \ncarc[arrows=->,linewidth=.03,arcangle=30]{42}{43}
        \ncarc[arrows=->,linewidth=.03,arcangle=-30]{12}{13}
        \ncarc[arrows=->,linewidth=.03,arcangle=-30]{13}{14}
        \ncarc[arrows=->,linewidth=.03,arcangle=-30]{31}{32}
        \ncarc[arrows=->,linewidth=.03,arcangle=-30]{34}{31}
        \ncarc[arrows=->,linewidth=.03,arcangle=30]{23}{24}
        \ncarc[arrows=->,linewidth=.03,arcangle=30]{24}{21}
        \ncarc[arrows=->,linewidth=.03,arcangle=-40]{42}{12}
        \ncarc[arrows=->,linewidth=.03,arcangle=-40]{32}{42}
        \ncarc[arrows=->,linewidth=.03,arcangle=40]{43}{13}
        \ncarc[arrows=->,linewidth=.03,arcangle=40]{31}{41}
        \ncarc[arrows=->,linewidth=.03,arcangle=40]{13}{23}
        \ncarc[arrows=->,linewidth=.03,arcangle=40]{21}{31}
        \ncarc[arrows=->,linewidth=.03,arcangle=-40]{14}{24}
        \ncarc[arrows=->,linewidth=.03,arcangle=-40]{24}{34}
      \end{psmatrix}
      \]
    \end{footnotesize}
    \caption{Conversion, preference and change of mind for the second
      version of the square game}
    \label{fig:conv_pref_square2}
  \end{figure}
}

\begin{document}

\title{Conversion/Preference Games}

\author{St\'ephane Le Roux   \and Pierre Lescanne \\
 Universit\'e de Lyon, CNRS (LIP), ENS de Lyon,\\
 46 all\'ee d'Italie, 69364 Lyon, France \\[2pt]
  \and Ren\'e Vestergaard \\
   School of Information Science, JAIST, 1-1 Asahidai, \\
   Nomi, Ishikawa 923-1292, Japan}
 
\date{}

 \maketitle

   \begin{abstract} We introduce the concept of
     \emph{Conversion/Preference Games}, or CP games for short. CP
     games generalize the standard notion of \emph{strategic games}.
     First we exemplify the use of CP games. Second we formally
     introduce and define the CP-games formalism.  Then we sketch two
     `real-life' applications, namely a connection between CP games
     and gene regulation networks, and the use of CP games to
     formalize implied information in Chinese Wall security.  We end
     with a study of a particular fixed-point construction over CP
     games and of the resulting existence of equilibria in possibly
     infinite games.
   \end{abstract}

\section{Introduction}
\label{sec:intro}

We give a stand-alone account of Conversion/Preference games or CP games, as originally used in~\cite{LeRouxLescanneVestergaard:RGT-Nash}.
CP games are built from a set of \emph{players} and
a set of (game) \emph{situations}.  The ability of the players to change a situation to another is formalised in \emph{conversion} relations.  A \emph{preference} relation dictates how the players compare the different situations against each other. The three main aims of this article are to show i) that discrete Nash-style game theory is possible and natural, ii) that the two basic CP concepts of \emph{Conversion} and \emph{Preference} are of wider interest, and iii) that game-theoretic notions are both applicable and relevant in situations where no payoff function need exist, or where the payoff concept would dramatically alter what aspects of the game are being considered.

\section{Basic concepts}
\label{sec:basic}

To start with, let us give the two main notions of games.  First, a
game involves \emph{players}. Second, a game is characterized by
situations.  In CP games these situations will be called
\emph{situations} or sometimes \emph{synopses} or \emph{game
  situations}.  A player can move from one situation to another, but
she\footnote{See the preface of~\cite{osborne04a} for the use of
  personal pronouns} does that under some constraints as she has no
total freedom to perform her moves, therefore a relation called
\emph{conversion} is defined for each player; it tells what moves a
player is allowed to perform.  Conversion of player \Alice{} will be
written $\conv{\Alice}$.  As such, conversion tells basically the
rules of the game.  In chess it would say \emph{``a~player can move
  her bishop along a diagonal''}, but it does not tell the game line
of the player.  In other words it does not tell why the player chooses
to move or to ``convert'' her situation.  Another relation called
\emph{preference} compares situations in order for a player to choose
a « better » move or to perform a « better » conversion.  Preference
of player \Beth{} will be written $\pref{\Beth}$ and when we write
$s\pref{\Beth}s'$ we mean that \Beth{} prefers $s'$ to~$s$ or, rather
than $s$, she chooses $s'$ or in situation~$s$ she is attracted toward
situation~$s'$.  Preference (or choice) is somewhat disconnected from
conversion, a player can clearly prefer a situation she cannot move to
and vice versa she can move to a situation she does not prefer.
Moreover players may share the same conversion relation, but this not
a rule and the may share or not the same preference relation or not.
Those different situation sill be illustrated by examples throughout
the article.

A key concept in games is this of \emph{equilibrium}.  As a player can
convert a situation, she can convert it into a situation she likes
better, in the sense that she prefers the new situation she converted
to.  A player is \emph{happy} in a situation, if there is no situation
she can convert into and she prefers.  A~situation is an equilibrium
if each player is happy with this situation.  We will see that this
concept of equilibrium captures and generalizes the concept known as
\emph{Nash equilibrium} in strategic games, hence the name
\emph{abstract Nash equilibrium}.

\section{Some examples}
\label{sec:examples}

Let us present the above concepts of conversion, preference and
equilibrium through examples.  We will introduce a new concept called
\emph{change of mind}.

\subsection{A simple game on a square}
\label{sec:square-game}

As an introduction, we will look at variations of a simple game on a
board.

\subsubsection{A first version}

Imagine a simple game where \Alice{} and \Beth{} play using tokens on
a square.  We number the four positions as $1$, $2$, $3$ and $4$.
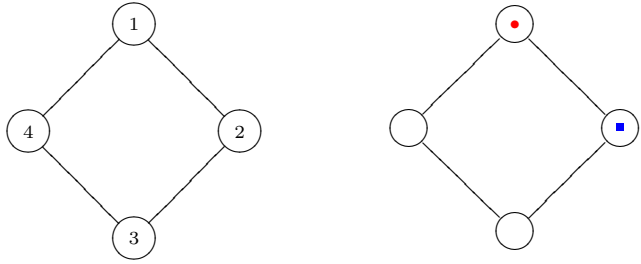
\begin{scriptsize}
  \[ \xymatrix { &*++[o][F-]{1}\ar@{-}[dl]\ar@{-}[dr]\\ *++[o][F-]{4}
\ar@{-}[dr]&& *++[o][F-]{2}\ar@{-}[dl]\\ &*++[o][F-]{3} } \qquad
\qquad \qquad %
\xymatrix {
&*++[o][F-]{\rouge{`(!)}}\ar@{-}[dl]\ar@{-}[dr]\\
*++[o][F-]{\phantom{`(!)}} \ar@{-}[dr]&&
*++[o][F-]{\bl{`[!]}}\ar@{-}[dl]\\ &*++[o][F-]{\phantom{`(!)}}  }
  \]
\end{scriptsize}

Assume that player \Alice{} has a red round token and that player
\Beth{} has a blue squared token.  The two players place their tokens
on vertices and then they move along edges.  They can also decide
not to move.  Assume that \Alice{} and \Beth{} never put their token
on a vertex taken by the other player and  a position further than this impossible situation is better than a position closer.  In other words, a position with Alice on vertex $i$ and Beth on vertex $j$ with $i-j$ even is preferred to a position with $i-j$ odd.

\begin{figure}[ht]
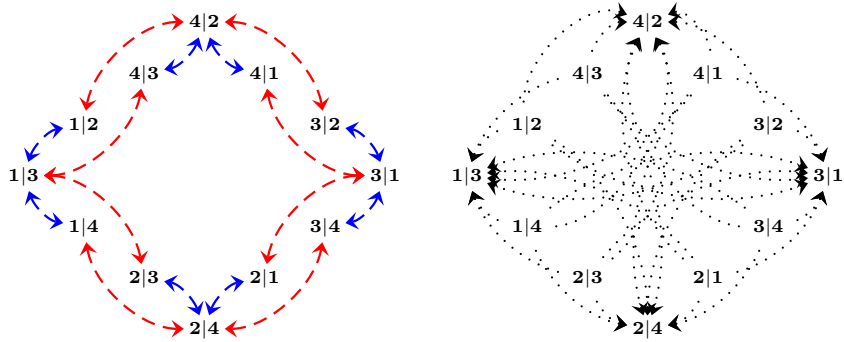

\noindent
\begin{scriptsize}
\(
  \begin{psmatrix}[colsep=.4cm,rowsep=.4cm]
  &&&& [name=42]{\mathbf{4|2}}\\
  &&& [name=43]{\mathbf{4|3}} && [name=41]{\mathbf{4|1}}\\
  && [name=12]{\mathbf{1|2}} &&&& [name=32]{\mathbf{3|2}} \\
  & [name=13]{\mathbf{1|3}} &&&&&& [name=31]{\mathbf{3|1}} \\
  && [name=14]{\mathbf{1|4}} &&&& [name=34]{\mathbf{3|4}} \\
  &&& [name=23]{\mathbf{2|3}} && [name=21]{\mathbf{2|1}} \\
  &&&& [name=24]{\mathbf{2|4}}
 \ncarc[arrows=<->,linewidth=.03,linestyle=dashed,linecolor=blue,arcangle=-30]{42}{41}
 \ncarc[arrows=<->,linewidth=.03,linestyle=dashed,linecolor=blue,arcangle=30]{42}{43}
 \ncarc[arrows=<->,linewidth=.03,linestyle=dashed,linecolor=blue,arcangle=-30]{12}{13}
 \ncarc[arrows=<->,linewidth=.03,linestyle=dashed,linecolor=blue,arcangle=-30]{13}{14}
 \ncarc[arrows=<->,linewidth=.03,linestyle=dashed,linecolor=blue,arcangle=30]{32}{31}
 \ncarc[arrows=<->,linewidth=.03,linestyle=dashed,linecolor=blue,arcangle=30]{31}{34}
 \ncarc[arrows=<->,linewidth=.03,linestyle=dashed,linecolor=blue,arcangle=30]{23}{24}
 \ncarc[arrows=<->,linewidth=.03,linestyle=dashed,linecolor=blue,arcangle=-30]{21}{24}
 \ncarc[arrows=<->,linewidth=.03,linestyle=dashed,linecolor=red,arcangle=-40]{42}{12}
 \ncarc[arrows=<->,linewidth=.03,linestyle=dashed,linecolor=red,arcangle=40]{42}{32}
 \ncarc[arrows=<->,linewidth=.03,linestyle=dashed,linecolor=red,arcangle=-40]{13}{43}
 \ncarc[arrows=<->,linewidth=.03,linestyle=dashed,linecolor=red,arcangle=40]{31}{41}
 \ncarc[arrows=<->,linewidth=.03,linestyle=dashed,linecolor=red,arcangle=40]{13}{23}
 \ncarc[arrows=<->,linewidth=.03,linestyle=dashed,linecolor=red,arcangle=-40]{31}{21}
 \ncarc[arrows=<->,linewidth=.03,linestyle=dashed,linecolor=red,arcangle=-40]{14}{24}
 \ncarc[arrows=<->,linewidth=.03,linestyle=dashed,linecolor=red,arcangle=40]{34}{24}
\end{psmatrix}
\quad
\begin{psmatrix}[colsep=.4cm,rowsep=.4cm]
  &&&& [name=42]{\mathbf{4|2}}\\
  &&& [name=43]{\mathbf{4|3}} && [name=41]{\mathbf{4|1}}\\
  && [name=12]{\mathbf{1|2}} &&&& [name=32]{\mathbf{3|2}} \\
  & [name=13]{\mathbf{1|3}} &&&&&& [name=31]{\mathbf{3|1}} \\
  && [name=14]{\mathbf{1|4}} &&&& [name=34]{\mathbf{3|4}} \\
  &&& [name=23]{\mathbf{2|3}} && [name=21]{\mathbf{2|1}} \\
  &&&& [name=24]{\mathbf{2|4}}
\ncarc[arrows=->,linewidth=.03,linestyle=dotted,arcangle=-30]{12}{13}
\ncarc[arrows=->,linewidth=.03,linestyle=dotted,arcangle=60]{12}{42}
\ncarc[arrows=->,linewidth=.03,linestyle=dotted,arcangle=-30]{12}{31}
\ncarc[arrows=->,linewidth=.03,linestyle=dotted,arcangle=30]{12}{24}
\ncarc[arrows=->,linewidth=.03,linestyle=dotted,arcangle=-30]{43}{13}
\ncarc[arrows=->,linewidth=.03,linestyle=dotted,arcangle=40]{43}{42}
\ncarc[arrows=->,linewidth=.03,linestyle=dotted,arcangle=-30]{43}{31}
\ncarc[arrows=->,linewidth=.03,linestyle=dotted,arcangle=30]{43}{24}
\ncarc[arrows=->,linewidth=.03,linestyle=dotted,arcangle=30]{41}{13}
\ncarc[arrows=->,linewidth=.03,linestyle=dotted,arcangle=-60]{41}{42}
\ncarc[arrows=->,linewidth=.03,linestyle=dotted,arcangle=30]{41}{31}
\ncarc[arrows=->,linewidth=.03,linestyle=dotted,arcangle=-30]{41}{24}
\ncarc[arrows=->,linewidth=.03,linestyle=dotted,arcangle=30]{32}{13}
\ncarc[arrows=->,linewidth=.03,linestyle=dotted,arcangle=-40]{32}{42}
\ncarc[arrows=->,linewidth=.03,linestyle=dotted,arcangle=30]{32}{31}
\ncarc[arrows=->,linewidth=.03,linestyle=dotted,arcangle=-30]{32}{24}
\ncarc[arrows=->,linewidth=.03,linestyle=dotted,arcangle=-30]{34}{13}
\ncarc[arrows=->,linewidth=.03,linestyle=dotted,arcangle=60]{34}{42}
\ncarc[arrows=->,linewidth=.03,linestyle=dotted,arcangle=-30]{34}{31}
\ncarc[arrows=->,linewidth=.03,linestyle=dotted,arcangle=30]{34}{24}
\ncarc[arrows=->,linewidth=.03,linestyle=dotted,arcangle=-30]{21}{13}
\ncarc[arrows=->,linewidth=.03,linestyle=dotted,arcangle=40]{21}{42}
\ncarc[arrows=->,linewidth=.03,linestyle=dotted,arcangle=-30]{21}{31}
\ncarc[arrows=->,linewidth=.03,linestyle=dotted,arcangle=30]{21}{24}
\ncarc[arrows=->,linewidth=.03,linestyle=dotted,arcangle=30]{14}{13}
\ncarc[arrows=->,linewidth=.03,linestyle=dotted,arcangle=-60]{14}{42}
\ncarc[arrows=->,linewidth=.03,linestyle=dotted,arcangle=30]{14}{31}
\ncarc[arrows=->,linewidth=.03,linestyle=dotted,arcangle=-30]{14}{24}
\ncarc[arrows=->,linewidth=.03,linestyle=dotted,arcangle=30]{23}{13}
\ncarc[arrows=->,linewidth=.03,linestyle=dotted,arcangle=-40]{23}{42}
\ncarc[arrows=->,linewidth=.03,linestyle=dotted,arcangle=30]{23}{31}
\ncarc[arrows=->,linewidth=.03,linestyle=dotted,arcangle=-30]{23}{24}
\end{psmatrix}
\)
\end{scriptsize}
\caption{Conversion and preference for the square game}
\label{fig:conv_pref_square}
\end{figure}

\begin{figure}[ht]
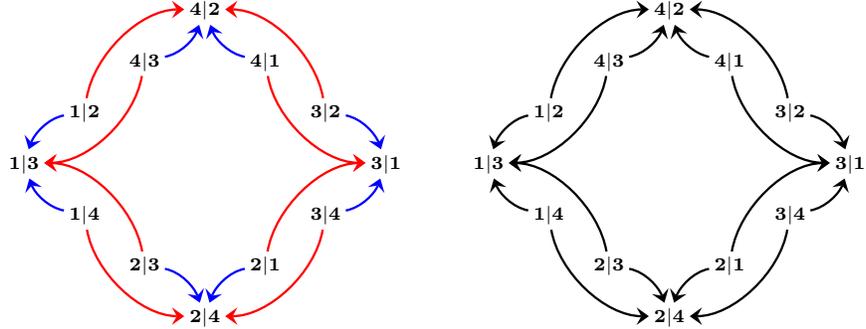

  \begin{scriptsize}
\[ 
\begin{psmatrix}[colsep=.4cm,rowsep=.4cm]
  &&&& [name=42]{\mathbf{4|2}}\\
  &&& [name=43]{\mathbf{4|3}} && [name=41]{\mathbf{4|1}}\\
  && [name=12]{\mathbf{1|2}} &&&& [name=32]{\mathbf{3|2}} \\
  & [name=13]{\mathbf{1|3}} &&&&&& [name=31]{\mathbf{3|1}} \\
  && [name=14]{\mathbf{1|4}} &&&& [name=34]{\mathbf{3|4}} \\
  &&& [name=23]{\mathbf{2|3}} && [name=21]{\mathbf{2|1}} \\
  &&&& [name=24]{\mathbf{2|4}}
  \ncarc[arrows=->,linewidth=.03,linecolor=blue,arcangle=-30]{12}{13}
 \ncarc[arrows=->,linewidth=.03,linecolor=blue,arcangle=30]{14}{13}
 \ncarc[arrows=->,linewidth=.03,linecolor=blue,arcangle=-30]{21}{24}
 \ncarc[arrows=->,linewidth=.03,linecolor=blue,arcangle=30]{23}{24}
 \ncarc[arrows=->,linewidth=.03,linecolor=blue,arcangle=30]{32}{31}
 \ncarc[arrows=->,linewidth=.03,linecolor=blue,arcangle=-30]{34}{31}
 \ncarc[arrows=->,linewidth=.03,linecolor=blue,arcangle=-30]{43}{42}
 \ncarc[arrows=->,linewidth=.03,linecolor=red,arcangle=40]{12}{42}
 \ncarc[arrows=->,linewidth=.03,linecolor=red,arcangle=-40]{14}{24}
 \ncarc[arrows=->,linewidth=.03,linecolor=red,arcangle=40]{21}{31}
 \ncarc[arrows=->,linewidth=.03,linecolor=red,arcangle=40]{43}{13}
 \ncarc[arrows=->,linewidth=.03,linecolor=red,arcangle=-40]{23}{13}
 \ncarc[arrows=->,linewidth=.03,linecolor=red,arcangle=-40]{32}{42}
 \ncarc[arrows=->,linewidth=.03,linecolor=red,arcangle=40]{34}{24}
 \ncarc[arrows=->,linewidth=.03,linecolor=red,arcangle=-40]{41}{31}
\ncarc[arrows=->,linewidth=.03,linecolor=blue,arcangle=30]{41}{42}
\end{psmatrix}
\qquad
\begin{psmatrix}[colsep=.4cm,rowsep=.4cm]
  &&&& [name=42]{\mathbf{4|2}}\\
  &&& [name=43]{\mathbf{4|3}} && [name=41]{\mathbf{4|1}}\\
  && [name=12]{\mathbf{1|2}} &&&& [name=32]{\mathbf{3|2}} \\
  & [name=13]{\mathbf{1|3}} &&&&&& [name=31]{\mathbf{3|1}} \\
  && [name=14]{\mathbf{1|4}} &&&& [name=34]{\mathbf{3|4}} \\
  &&& [name=23]{\mathbf{2|3}} && [name=21]{\mathbf{2|1}} \\
  &&&& [name=24]{\mathbf{2|4}}
  \ncarc[arrows=->,linewidth=.03,arcangle=-30]{12}{13}
 \ncarc[arrows=->,linewidth=.03,arcangle=30]{14}{13}
 \ncarc[arrows=->,linewidth=.03,arcangle=-30]{21}{24}
 \ncarc[arrows=->,linewidth=.03,arcangle=30]{23}{24}
 \ncarc[arrows=->,linewidth=.03,arcangle=30]{32}{31}
 \ncarc[arrows=->,linewidth=.03,arcangle=-30]{34}{31}
 \ncarc[arrows=->,linewidth=.03,arcangle=-30]{43}{42}
 \ncarc[arrows=->,linewidth=.03,arcangle=40]{12}{42}
 \ncarc[arrows=->,linewidth=.03,arcangle=-40]{14}{24}
 \ncarc[arrows=->,linewidth=.03,arcangle=40]{21}{31}
 \ncarc[arrows=->,linewidth=.03,arcangle=40]{43}{13}
 \ncarc[arrows=->,linewidth=.03,arcangle=-40]{23}{13}
 \ncarc[arrows=->,linewidth=.03,arcangle=-40]{32}{42}
 \ncarc[arrows=->,linewidth=.03,arcangle=40]{34}{24}
 \ncarc[arrows=->,linewidth=.03,arcangle=-40]{41}{31}
\ncarc[arrows=->,linewidth=.03,arcangle=30]{41}{42}
\end{psmatrix}
\]
\end{scriptsize}
  \caption{Agent changes of mind (on the left) and (general) change of
mind (on the right) for the square game}
  \label{fig:chg_of_mind_square}
\end{figure}

The game has 12 situations, which we write $i|j$ for
${1\le i,j\le 4}$ and $i\not=j$.  The above pictured situation
corresponds to $1|2$.  The two conversions are described by
Figure~\ref{fig:conv_pref_square} left.  In this figure,
$\rouge{\dashrightarrow}$ is \Alice's conversion and $\bl{\dashrightarrow}$
is \Beth's conversion.

In this game, both players share the same preference, namely the
following: since a player does not want her token on a position next to the other token, she
prefers a situation where her token is on the opposite corner of the
other token.
This gives the preference given in Figure~\ref{fig:conv_pref_square}
right.  The arrow from $1|2$ to $1|3$ means players prefer $1|3$ to
$1|2$.

From the conversion and the preference we build a relation that we
call \emph{change of mind}.  \Alice{} can change her mind from a
situation $s$ to a new one $s'$, if she can convert $s$ into the new
situation $s'$ and rather than $s$ she chooses~$s'$.  Changes of mind
for \Alice{} and \Beth{} are given in
Figure~\ref{fig:chg_of_mind_square} on the left.  In this figure,
$\rouge{\xymatrix{\ar@{->}[r]&}}$ is \Alice's change of mind and
$\bl{\xymatrix{\ar@{->}[r]&}}$ is \Beth's conversion.  The
\emph{(general) change of mind} is the union of the \emph{agent change
  of mind}, it is given by Figure~\ref{fig:chg_of_mind_square} on the
right.  The equilibria are the end points (or ``minimal point'') for
that relation, namely $1|3$, $4|2$, $3|1$ and $2|4$.  This means that
no change of mind arrows leave those nodes.  In these situations
players have their tokens on opposite corners and they do not move.
An equilibrium like $1|3$ which is an end point is called an
\emph{Abstract Nash Equilibrium}.

\gametwo

\subsubsection{A second version}

We propose a second version of the game, where moves of the token can
only be made clockwise.  This implies to change the conversion
changes, but also the preference, as a player does want not to be
threatened by another token placed before hers clockwise and prefers a
situation that places this token as far as possible.  The conversions,
the preferences and the changes of mind are given in
Figure~\ref{fig:conv_pref_square2}
(page~\pageref{fig:conv_pref_square2}).  If one looks at the
equilibrium, one sees that there is no fixed position where players
are happy.  To be happy the players have to move around for ever, one
chasing the other.  It is not really a cycle, but a perpetual move.
We also call that an equilibrium.  It is sometimes called a
\emph{dynamic equilibrium} or a \emph{stationary state}.

\subsubsection{A third version}

The third version is meant to present an interesting feature of the
change of mind.  In this version, we use the same rules as the second
one, except that we suppose that the game does not start with both
token on the board.  Actually it starts as follows.  \Alice{} has put
her token on node $1$ (this game positions is described as $1|`w$).
Then \Beth{} chooses a position among $2$, $3$ or~$4$.  The conversion
is given in Figure~\ref{fig:conv_CoM_square3} left.  \Beth{} may
choose not to play, but in this case she loses, in other words, she
prefers any position to $1|`w$.  We do not draw the preference relation, as it would make for an entangled picture. The change of mind is given on
Figure~\ref{fig:conv_CoM_square3} right
(page~\pageref{fig:conv_CoM_square3}).  There is again a dynamic
equilibrium and one sees that this dynamic equilibrium is not the
whole game, indeed one enters the perpetual move after at least one
step in the game.  

\begin{figure}[bt]
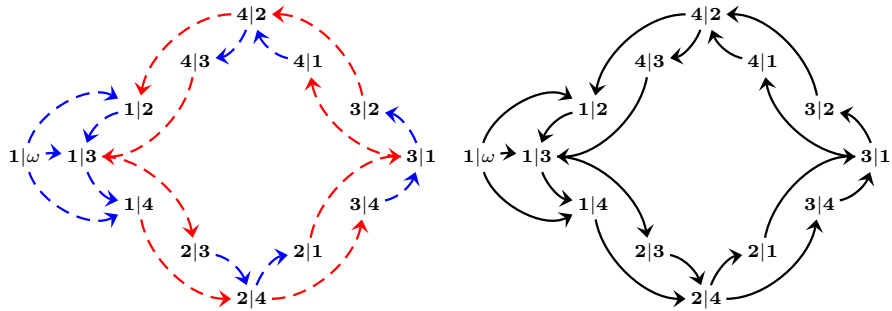

  \begin{scriptsize}
\noindent\(
\begin{psmatrix}[colsep=.35cm,rowsep=.35cm]
  &&&&& [name=42]{\mathbf{4|2}}\\
  &&&& [name=43]{\mathbf{4|3}} && [name=41]{\mathbf{4|1}}\\
  &&& [name=12]{\mathbf{1|2}} &&&& [name=32]{\mathbf{3|2}} \\
  &[name=1w]{\mathbf{1|\omega}}& [name=13]{\mathbf{1|3}} &&&&&& [name=31]{\mathbf{3|1}} \\
  &&& [name=14]{\mathbf{1|4}} &&&& [name=34]{\mathbf{3|4}} \\
  &&&& [name=23]{\mathbf{2|3}} && [name=21]{\mathbf{2|1}} \\
  &&&&& [name=24]{\mathbf{2|4}}
 \ncarc[arrows=->,linewidth=.03,linestyle=dashed,linecolor=blue,arcangle=50]{1w}{12}
 \ncarc[arrows=->,linewidth=.03,linestyle=dashed,linecolor=blue]{1w}{13}
 \ncarc[arrows=->,linewidth=.03,linestyle=dashed,linecolor=blue,arcangle=-50]{1w}{14}
 \ncarc[arrows=->,linewidth=.03,linestyle=dashed,linecolor=blue,arcangle=30]{41}{42}
 \ncarc[arrows=->,linewidth=.03,linestyle=dashed,linecolor=blue,arcangle=30]{42}{43}
 \ncarc[arrows=->,linewidth=.03,linestyle=dashed,linecolor=blue,arcangle=-30]{12}{13}
 \ncarc[arrows=->,linewidth=.03,linestyle=dashed,linecolor=blue,arcangle=-30]{13}{14}
 \ncarc[arrows=->,linewidth=.03,linestyle=dashed,linecolor=blue,arcangle=-30]{31}{32}
 \ncarc[arrows=->,linewidth=.03,linestyle=dashed,linecolor=blue,arcangle=-30]{34}{31}
 \ncarc[arrows=->,linewidth=.03,linestyle=dashed,linecolor=blue,arcangle=30]{23}{24}
 \ncarc[arrows=->,linewidth=.03,linestyle=dashed,linecolor=blue,arcangle=30]{24}{21}
 \ncarc[arrows=->,linewidth=.03,linestyle=dashed,linecolor=red,arcangle=-40]{42}{12}
 \ncarc[arrows=->,linewidth=.03,linestyle=dashed,linecolor=red,arcangle=-40]{32}{42}
 \ncarc[arrows=->,linewidth=.03,linestyle=dashed,linecolor=red,arcangle=40]{43}{13}
 \ncarc[arrows=->,linewidth=.03,linestyle=dashed,linecolor=red,arcangle=40]{31}{41}
 \ncarc[arrows=->,linewidth=.03,linestyle=dashed,linecolor=red,arcangle=40]{13}{23}
 \ncarc[arrows=->,linewidth=.03,linestyle=dashed,linecolor=red,arcangle=40]{21}{31}
 \ncarc[arrows=->,linewidth=.03,linestyle=dashed,linecolor=red,arcangle=-40]{14}{24}
 \ncarc[arrows=->,linewidth=.03,linestyle=dashed,linecolor=red,arcangle=-40]{24}{34}
\end{psmatrix}
\begin{psmatrix}[colsep=.35cm,rowsep=.35cm]
  &&&&& [name=42]{\mathbf{4|2}}\\
  &&&& [name=43]{\mathbf{4|3}} && [name=41]{\mathbf{4|1}}\\
  &&& [name=12]{\mathbf{1|2}} &&&& [name=32]{\mathbf{3|2}} \\
  &[name=1w]{\mathbf{1|\omega}}& [name=13]{\mathbf{1|3}} &&&&&& [name=31]{\mathbf{3|1}} \\
  &&& [name=14]{\mathbf{1|4}} &&&& [name=34]{\mathbf{3|4}} \\
  &&&& [name=23]{\mathbf{2|3}} && [name=21]{\mathbf{2|1}} \\
  &&&&& [name=24]{\mathbf{2|4}}
 \ncarc[arrows=->,linewidth=.03,arcangle=50]{1w}{12}
 \ncarc[arrows=->,linewidth=.03]{1w}{13}
 \ncarc[arrows=->,linewidth=.03,arcangle=-50]{1w}{14}
 \ncarc[arrows=->,linewidth=.03,arcangle=30]{41}{42}
 \ncarc[arrows=->,linewidth=.03,arcangle=30]{42}{43}
 \ncarc[arrows=->,linewidth=.03,arcangle=-30]{12}{13}
 \ncarc[arrows=->,linewidth=.03,arcangle=-30]{13}{14}
 \ncarc[arrows=->,linewidth=.03,arcangle=-30]{31}{32}
 \ncarc[arrows=->,linewidth=.03,arcangle=-30]{34}{31}
 \ncarc[arrows=->,linewidth=.03,arcangle=30]{23}{24}
 \ncarc[arrows=->,linewidth=.03,arcangle=30]{24}{21}
 \ncarc[arrows=->,linewidth=.03,arcangle=-40]{42}{12}
 \ncarc[arrows=->,linewidth=.03,arcangle=-40]{32}{42}
 \ncarc[arrows=->,linewidth=.03,arcangle=40]{43}{13}
 \ncarc[arrows=->,linewidth=.03,arcangle=40]{31}{41}
 \ncarc[arrows=->,linewidth=.03,arcangle=40]{13}{23}
 \ncarc[arrows=->,linewidth=.03,arcangle=40]{21}{31}
 \ncarc[arrows=->,linewidth=.03,arcangle=-40]{14}{24}
 \ncarc[arrows=->,linewidth=.03,arcangle=-40]{24}{34}
\end{psmatrix}
\)
\end{scriptsize}
\caption{Conversion and change of mind for the third version of the square
game}
\label{fig:conv_CoM_square3}
\end{figure}

\subsection{Strategic games}
\label{sec:strat_games}

In this presentation of strategic games we do not use payoff
functions, but directly a preference relation (See Section
1.1.2 of~\cite{osborne04a} for a discussion) and we present several
games.

\subsubsection{The Prisoner's Dilemma} The problem is stated usually
as follows

\begin{quote} Two suspects, A and B, are arrested by the police. The
police have insufficient evidence for a conviction, and, having
separated both prisoners, visit each of them to offer the same deal:
if one acts as an informer against the other (\emph{finks}) and the
other remains \emph{quiet}, the betrayer goes free and the quiet
accomplice receives the full sentence. If both stay quiet, the police
can sentence both prisoners to a reduced sentence in jail for a minor
charge. If each finks, each will receive a similar intermediate
sentence.  Each prisoner must make the choice of whether to fink or to
remain quiet. However, neither prisoner knows for sure what choice the
other prisoner will make. So the question this dilemma poses is: What
will happen? How will the prisoners act?
\end{quote}

Each prisoner can be into two states, either \emph{fink} ($F$) or be
\emph{quiet} ($Q$).  Each prisoner can go from $Q$ to $F$ and
vice-versa, hence the following conversion, where
$\ConvA$ is prisoner A conversion and
$\ConvB$ is prisoner B conversion
(Figure~\ref{fig:prisoner_conv_pref} left).  Each prisoner prefers to
go free over being sentenced and prefers a light sentence to a full
sentence.  Hence the preference are given in
Figure~\ref{fig:prisoner_conv_pref} right, where
$\prefAlice$ is prisoner A preference and
$\prefBeth$ is prisoner B preference.
\begin{figure}[ht]
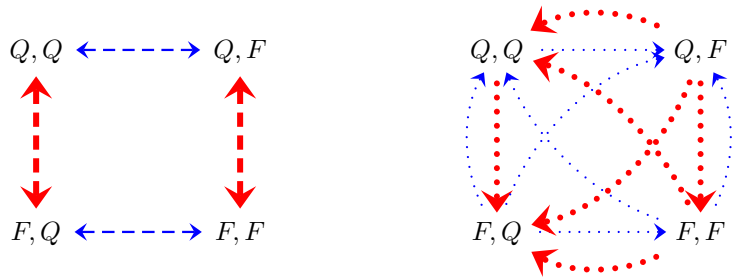

  \[
  \begin{psmatrix}[colsep=2cm,rowsep=2cm]
    &[name=QQ]{Q,Q} & [name=QF]{Q,F}\\
    &[name=FQ]{F,Q} & [name=FF]{F,F} \psset{nodesep=5pt}
    \ncline[arrows=<->,linewidth=.03,linecolor=blue,linestyle=dashed]{QQ}{QF}
    \ncline[arrows=<->,linewidth=.03,linecolor=blue,linestyle=dashed]{FQ}{FF}
    \ncline[arrows=<->,linewidth=.08,linecolor=red,linestyle=dashed]{QQ}{FQ}
    \ncline[arrows=<->,linewidth=.08,linecolor=red,linestyle=dashed]{QF}{FF}
  \end{psmatrix}
  \qquad
  \begin{psmatrix}[colsep=2cm,rowsep=2cm]
    &[name=QQ]{Q,Q} & [name=QF]{Q,F}\\
    &[name=FQ]{F,Q} & [name=FF]{F,F} \psset{nodesep=3.5pt}
    \ncline[arrows=->,linewidth=.03,linecolor=blue,linestyle=dotted]{QQ}{QF}
    \ncarc[arrows=->,linewidth=.08,linecolor=red,linestyle=dotted,
    arcangle=-30]{QF}{QQ}
    \ncline[arrows=->,linewidth=.03,linecolor=blue,linestyle=dotted]{FQ}{FF}
    \ncarc[arrows=->,linewidth=.08,linecolor=red,linestyle=dotted,
    arcangle=30]{FF}{FQ}
    \ncline[arrows=->,linewidth=.08,linecolor=red,linestyle=dotted]{QQ}{FQ}
    \ncarc[arrows=->,linewidth=.03,linecolor=blue,linestyle=dotted,
    arcangle=30]{FQ}{QQ}
    \ncline[arrows=->,linewidth=.08,linecolor=red,linestyle=dotted]{QF}{FF}
    \ncarc[arrows=->,linewidth=.03,linecolor=blue,linestyle=dotted,
    arcangle=-30]{FF}{QF}
    \ncarc[arrows=->,linewidth=.08,linecolor=red,linestyle=dotted,
    arcangle=-25]{FF}{QQ}
    \ncarc[arrows=->,linewidth=.03,linecolor=blue,linestyle=dotted,
    arcangle=25]{FF}{QQ}
    \ncarc[arrows=->,linewidth=.08,linecolor=red,linestyle=dotted,
    arcangle=30]{QF}{FQ}
    \ncarc[arrows=->,linewidth=.03,linecolor=blue,linestyle=dotted,
    arcangle=30]{FQ}{QF}
  \end{psmatrix}
  \qquad
  \]

  \bigskip

  \caption{Conversion and preference in the prisoner's dilemma}
  \label{fig:prisoner_conv_pref}
\end{figure}

From this we get the change of mind of Figure~\ref{fig:prisoner_CoM}.
One sees clearly that the only equilibrium is $F,F$ despite both
prefer $Q,Q$ as shown on Figure~\ref{fig:prisoner_conv_pref} right.
\begin{figure}[ht]
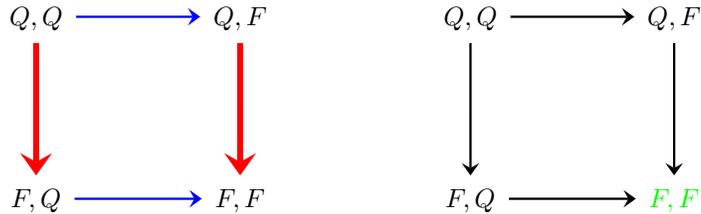
 \medskip
\[ 
\begin{psmatrix}[colsep=2cm,rowsep=2cm]
&[name=QQ]{Q,Q} & [name=QF]{Q,F}\\
&[name=FQ]{F,Q} & [name=FF]{F,F}
\psset{nodesep=5pt}
\ncline[arrows=->,linewidth=.03,linecolor=blue]{QQ}{QF}
\ncline[arrows=->,linewidth=.03,linecolor=blue]{FQ}{FF}
\ncline[arrows=->,linewidth=.08,linecolor=red]{QQ}{FQ}
\ncline[arrows=->,linewidth=.08,linecolor=red]{QF}{FF}
\end{psmatrix}
\quad
\begin{psmatrix}[colsep=2cm,rowsep=2cm]
&[name=QQ]{Q,Q} & [name=QF]{Q,F}\\
&[name=FQ]{F,Q} & [name=FF]{\color{green}F,F}
\psset{nodesep=5pt}
\ncline[arrows=->,linewidth=.03]{QQ}{QF}
\ncline[arrows=->,linewidth=.03]{FQ}{FF}
\ncline[arrows=->,linewidth=.03]{QQ}{FQ}
\ncline[arrows=->,linewidth=.03]{QF}{FF}
\end{psmatrix}
\]
\caption{Agent and (general) change of mind in the prisoner's dilemma}
\label{fig:prisoner_CoM} 
\end{figure}

Such an equilibrium is called a Nash equilibrium in strategic game
theory.

The paradox comes from the fact that $F,F$ is an equilibrium despite
the fact one has:
\begin{scriptsize}
  $
  \begin{psmatrix}[colsep=1.5cm,rowsep=1.5cm] %
    [name=QQ]{Q,Q} &
    [name=FF]{F,F} 
    \psset{nodesep=3.5pt}
    \ncarc[arrows=->,linewidth=.06,linecolor=red,linestyle=dotted, arcangle=-15]{FF}{QQ}
    \ncarc[arrows=->,linewidth=.025,linecolor=blue,linestyle=dotted, arcangle=15]{FF}{QQ}
  \end{psmatrix}
  $
\end{scriptsize}
in the preference.

\subsubsection{Matching Pennies}

This second example is also classic.  This is a simple example of
strategic game where there is no singleton equilibrium.  As an
equilibrium can contain more than one situation, we call singleton
equilibrium a CP equilibrium which contains only one situation.  This
boils down to the kind of equilibrium we have introduced so far.

\begin{quotation}
  \noindent The game is played between two players, Player A and
Player~B. Each player has a penny and must secretly turn the penny to
heads ($H$) or tails ($T$). The players then reveal their choices
simultaneously. If the pennies match (both heads or both tails),
Player A wins. If the pennies do not match (one heads and one tails),
Player B wins.
\end{quotation}

The conversion is similar to this of the prisoner's dilemma
(Figure~\ref{fig:match_pen_conv_pref_CoM} left) and the preference is
given by who wins (Figure~\ref{fig:match_pen_conv_pref_CoM} center).
\begin{figure}[ht]
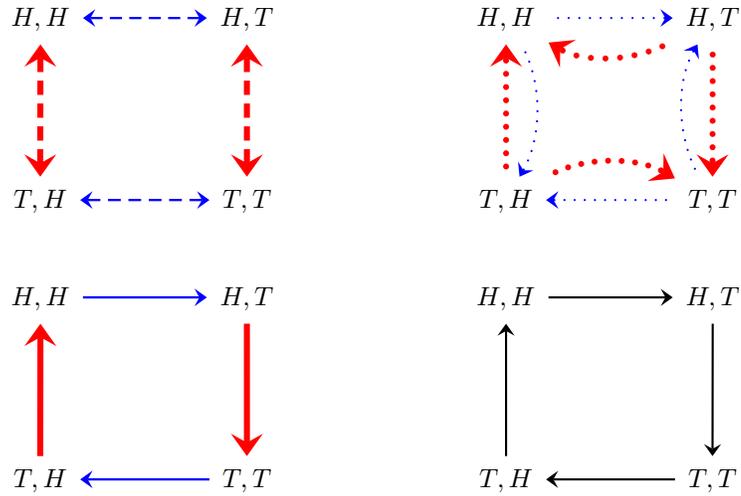
 
\[
\begin{psmatrix}[colsep=2cm,rowsep=2cm]
  & [name=HH]{H,H} & [name=HT]{H,T}\\
  & [name=TH]{T,H} & [name=TT]{T,T}
  \psset{nodesep=5pt}
  \ncline[arrows=<->,linewidth=.03,linecolor=blue,linestyle=dashed]{HH}{HT}
  \ncline[arrows=<->,linewidth=.03,linecolor=blue,linestyle=dashed]{TH}{TT}
  \ncline[arrows=<->,linewidth=.08,linecolor=red,linestyle=dashed]{HH}{TH}
  \ncline[arrows=<->,linewidth=.08,linecolor=red,linestyle=dashed]{HT}{TT}
\end{psmatrix}
\qquad
\begin{psmatrix}[colsep=2cm,rowsep=2cm]
  & [name=HH]{H,H} & [name=HT]{H,T}\\
  & [name=TH]{T,H} & [name=TT]{T,T}
  \psset{nodesep=5pt}
  \ncline[arrows=->,linewidth=.03,linecolor=blue,linestyle=dotted]{HH}{HT}
  \ncline[arrows=->,linewidth=.03,linecolor=blue,linestyle=dotted]{TT}{TH}
  \ncline[arrows=->,linewidth=.08,linecolor=red,linestyle=dotted]{TH}{HH}
  \ncline[arrows=->,linewidth=.08,linecolor=red,linestyle=dotted]{HT}{TT}
  \ncarc[arrows=->,linewidth=.08,linecolor=red,linestyle=dotted,arcangle=30]{HT}{HH}
  \ncarc[arrows=->,linewidth=.03,linecolor=blue,linestyle=dotted,arcangle=30]{HH}{TH}
  \ncarc[arrows=->,linewidth=.08,linecolor=red,linestyle=dotted,arcangle=30]{TH}{TT}
  \ncarc[arrows=->,linewidth=.03,linecolor=blue,linestyle=dotted,arcangle=30]{TT}{HT}
\end{psmatrix}
\]

\[
\begin{psmatrix}[colsep=2cm,rowsep=2cm]
  & [name=HH]{H,H} & [name=HT]{H,T}\\
  & [name=TH]{T,H} & [name=TT]{T,T}
  \psset{nodesep=5pt}
  \ncline[arrows=->,linewidth=.03,linecolor=blue]{HH}{HT}
  \ncline[arrows=->,linewidth=.03,linecolor=blue]{TT}{TH}
  \ncline[arrows=->,linewidth=.08,linecolor=red]{TH}{HH}
  \ncline[arrows=->,linewidth=.08,linecolor=red]{HT}{TT}
\end{psmatrix}
\qquad
\begin{psmatrix}[colsep=2cm,rowsep=2cm]
  & [name=HH]{H,H} & [name=HT]{H,T}\\
  & [name=TH]{T,H} & [name=TT]{T,T}
  \psset{nodesep=5pt}
  \ncline[arrows=->]{HH}{HT}
  \ncline[arrows=->]{TT}{TH}
  \ncline[arrows=->]{TH}{HH}
  \ncline[arrows=->]{HT}{TT}
\end{psmatrix}
\]
\caption{Conversion, preference and (general) change of mind in
Matching Pennies}
\label{fig:match_pen_conv_pref_CoM}
\end{figure}

Change of mind for matching pennies is in
Figure~\ref{fig:match_pen_conv_pref_CoM} right.  One notices that
there is a cycle.  This cycle is the equilibrium.  No player has clear
mind of what to play and changes her minds each time she loses.

\subsubsection{Scissors, Paper, Stone}
\label{sec:sc_pap_st}

Here we present the famous game known as \emph{scissors, paper,
stone}.  It involves two players, \Alice{} and \Beth{} who announce
either \emph{scissors} ($C$) or \emph{paper} ($P$) or \emph{stone}
($T$) with the rules that \emph{stone beats scissors, scissors beat
paper, and paper beats stone}.  There are nine situations (see below),
one sees that \Alice{} may convert her situation $C,P$ into $P,P$ or
$T,P$ and the same for the other situations.  The conversion is given
below left.  Since the rules, it seems clear that \Alice{} prefers
$P,P$ to $T,P$ and $C,P$ to $P,P$, hence the preference given below
right with $\prefAlice$ is \Alice{}'s preference and
$\prefBeth$ is \Beth{}'s preference.  To avoid a
cumbersome diagram, in the preference we do not put the arrows deduced
by transitivity.

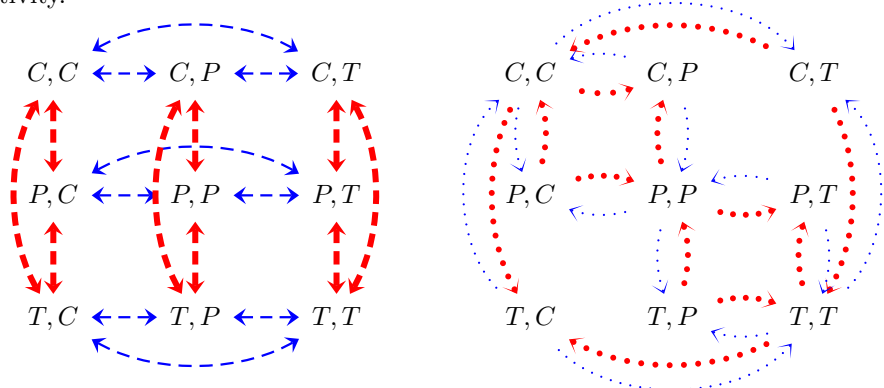
\[
\begin{psmatrix}[colsep=1.2cm,rowsep=1.2cm]
  &[name=CC]{C,C} & [name=CP]{C,P} &  [name=CT]{C,T} \\
  &[name=PC]{P,C} & [name=PP]{P,P} &  [name=PT]{P,T} \\
  &[name=TC]{T,C} & [name=TP]{T,P} &  [name=TT]{T,T}
\psset{nodesep=5pt} 
\ncline[arrows=<->,linewidth=.03,linecolor=blue,linestyle=dashed]{CC}{CP}
\ncline[arrows=<->,linewidth=.03,linecolor=blue,linestyle=dashed]{CP}{CT}
\ncarc[arrows=<->,linewidth=.03,linecolor=blue,linestyle=dashed,arcangle=30]{CC}{CT}
\ncline[arrows=<->,linewidth=.03,linecolor=blue,linestyle=dashed]{PC}{PP}
\ncline[arrows=<->,linewidth=.03,linecolor=blue,linestyle=dashed]{PP}{PT}
\ncarc[arrows=<->,linewidth=.03,linecolor=blue,linestyle=dashed,arcangle=30]{PC}{PT}
\ncline[arrows=<->,linewidth=.03,linecolor=blue,linestyle=dashed]{TC}{TP}
\ncline[arrows=<->,linewidth=.03,linecolor=blue,linestyle=dashed]{TP}{TT}
\ncarc[arrows=<->,linewidth=.03,linecolor=blue,linestyle=dashed,arcangle=-30]{TC}{TT}
\psset{arrowscale=.6}
\ncline[arrows=<->,linewidth=.08,linecolor=red,linestyle=dashed]{CC}{PC}
\ncline[arrows=<->,linewidth=.08,linecolor=red,linestyle=dashed]{PC}{TC}
\ncarc[arrows=<->,linewidth=.08,linecolor=red,linestyle=dashed,arcangle=-30]{CC}{TC}
\ncline[arrows=<->,linewidth=.08,linecolor=red,linestyle=dashed]{CP}{PP}
\ncline[arrows=<->,linewidth=.08,linecolor=red,linestyle=dashed]{PP}{TP}
\ncarc[arrows=<->,linewidth=.08,linecolor=red,linestyle=dashed,arcangle=-30]{CP}{TP}
\ncline[arrows=<->,linewidth=.08,linecolor=red,linestyle=dashed]{CT}{PT}
\ncline[arrows=<->,linewidth=.08,linecolor=red,linestyle=dashed]{PT}{TT}
\ncarc[arrows=<->,linewidth=.08,linecolor=red,linestyle=dashed,arcangle=30]{CT}{TT}
\end{psmatrix}
\qquad
\begin{psmatrix}[colsep=1.2cm,rowsep=1.2cm]
  &[name=CC]{C,C} & [name=CP]{C,P} &  [name=CT]{C,T} \\
  &[name=PC]{P,C} & [name=PP]{P,P} &  [name=PT]{P,T} \\
  &[name=TC]{T,C} & [name=TP]{T,P} &  [name=TT]{T,T}
\psset{nodesep=5pt,arrowinset=.8}
\ncarc[arrows=->,linewidth=.03,linecolor=blue,linestyle=dotted,arcangle=50]{CC}{CT}
\ncarc[arrows=->,linewidth=.08,arrowscale=.6,linecolor=red,linestyle=dotted,arcangle=-30]{CT}{CC}
\ncarc[arrows=->,linewidth=.03,linecolor=blue,linestyle=dotted,arcangle=-20]{CP}{CC}
\ncarc[arrows=->,linewidth=.08,arrowscale=.6,linecolor=red,linestyle=dotted,arcangle=-20]{CC}{CP}
\ncarc[arrows=->,linewidth=.03,linecolor=blue,linestyle=dotted,arcangle=20]{PP}{PC}
\ncarc[arrows=->,linewidth=.08,arrowscale=.6,linecolor=red,linestyle=dotted,arcangle=20]{PC}{PP}
\ncarc[arrows=->,linewidth=.03,linecolor=blue,linestyle=dotted,arcangle=-20]{PT}{PP}
\ncarc[arrows=->,linewidth=.08,arrowscale=.6,linecolor=red,linestyle=dotted,arcangle=-20]{PP}{PT}
\ncarc[arrows=->,linewidth=.03,linecolor=blue,linestyle=dotted,arcangle=-50]{TC}{TT}
\ncarc[arrows=->,linewidth=.08,arrowscale=.6,linecolor=red,linestyle=dotted,arcangle=30]{TT}{TC}
\ncarc[arrows=->,linewidth=.03,linecolor=blue,linestyle=dotted,arcangle=20]{TT}{TP}
\ncarc[arrows=->,linewidth=.08,arrowscale=.6,linecolor=red,linestyle=dotted,arcangle=20]{TP}{TT}
\ncarc[arrows=->,linewidth=.03,linecolor=blue,linestyle=dotted,arcangle=50]{TC}{CC}
\ncarc[arrows=->,linewidth=.08,arrowscale=.6,linecolor=red,linestyle=dotted,arcangle=-30]{CC}{TC}
\ncarc[arrows=->,linewidth=.03,linecolor=blue,linestyle=dotted,arcangle=-20]{CC}{PC}
\ncarc[arrows=->,linewidth=.08,arrowscale=.6,linecolor=red,linestyle=dotted,arcangle=-20]{PC}{CC}
\ncarc[arrows=->,linewidth=.03,linecolor=blue,linestyle=dotted,arcangle=20]{CP}{PP}
\ncarc[arrows=->,linewidth=.08,arrowscale=.6,linecolor=red,linestyle=dotted,arcangle=20]{PP}{CP}
\ncarc[arrows=->,linewidth=.03,linecolor=blue,linestyle=dotted,arcangle=-20]{PP}{TP}
\ncarc[arrows=->,linewidth=.08,arrowscale=.6,linecolor=red,linestyle=dotted,arcangle=-20]{TP}{PP}
\ncarc[arrows=->,linewidth=.03,linecolor=blue,linestyle=dotted,arcangle=-50]{TT}{CT}
\ncarc[arrows=->,linewidth=.08,arrowscale=.6,linecolor=red,linestyle=dotted,arcangle=30]{CT}{TT}
\ncarc[arrows=->,linewidth=.03,linecolor=blue,linestyle=dotted,arcangle=20]{PT}{TT}
\ncarc[arrows=->,linewidth=.08,arrowscale=.6,linecolor=red,linestyle=dotted,arcangle=20]{TT}{PT}
\end{psmatrix}
\]

\bigskip

From the above conversion and preference, one gets the following
change of mind.
\bigskip
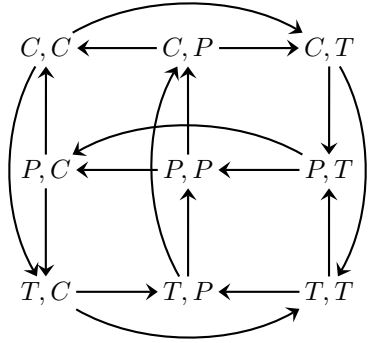
\[ 
\begin{psmatrix}[colsep=1.2cm,rowsep=1.2cm]
  &[name=CC]{C,C} & [name=CP]{C,P} &  [name=CT]{C,T} \\
  &[name=PC]{P,C} & [name=PP]{P,P} &  [name=PT]{P,T} \\
  &[name=TC]{T,C} & [name=TP]{T,P} &  [name=TT]{T,T}
\psset{nodesep=5pt,arrowinset=.8}
\end{psmatrix}
\ncline[arrows=->]{CP}{CC}
\ncline[arrows=->]{CP}{CT}
\ncarc[arrows=->,arcangle=30]{CC}{CT}
\ncline[arrows=->]{PT}{PP}
\ncline[arrows=->]{PP}{PC}
\ncarc[arrows=->,arcangle=-30]{PT}{PC}
\ncline[arrows=->]{TC}{TP}
\ncline[arrows=->]{TT}{TP}
\ncarc[arrows=->,arcangle=-30]{TC}{TT}
\ncline[arrows=->]{PC}{CC}
\ncline[arrows=->]{PC}{TC}
\ncarc[arrows=->,arcangle=-30]{CC}{TC}
\ncline[arrows=->]{TP}{PP}
\ncline[arrows=->]{PP}{CP}
\ncarc[arrows=->,arcangle=30]{TP}{CP}
\ncline[arrows=->]{CT}{PT}
\ncline[arrows=->]{TT}{PT}
\ncarc[arrows=->,arcangle=30]{CT}{TT}
\]

\bigskip

One sees also perpetual moves as in the matching pennies of which it
is a generalization.

\subsubsection{Strategic games as CP games}

A strategic game is a specific kind of CP games.  To be a strategic
game, a CP game has to fulfill the following conditions.

\begin{enumerate}
\item Each situation is a $n$-Cartesian product, where $n$ is the
number of players.  The constituents of the Cartesian product are
called \emph{strategies}.

\item Conversion for player $a$, written $\conv{a}$, is any change
along the $a$-th dimension, i.e., $(s_1,...,s_a,...,s_n) \conv{a}
(s_1,...,s_a',...,s_n)$.  Hence in
strategic games, conversion is an equivalent relation, namely
  \begin{itemize}
  \item symmetric, ($s \conv{a} s'$ implies $s' \conv{a} s$),
  \item transitive, ($s \conv{a} s'$ and $s' \conv{a} s''$ imply $s
\conv{a} s''$),
  \item and reflexive ($s \conv{a} s$).
  \end{itemize}

\end{enumerate}

\subsection{Blink and you lose}
\label{sec:blink} \emph{Blink and you lose} is a game played on a
simple graph with two undifferentiated tokens.  There are three
positions:
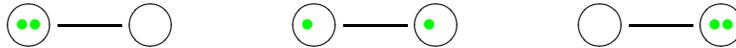
\[ \xymatrix { *++[o][F-]{\verte{`(!)}\verte{`(!)}}\ar@{-}[r] &
*++[o][F-]{\phantom{`(!)}\phantom{`(!)}}  } \qquad\qquad \xymatrix {
*++[o][F-]{\verte{`(!)}\phantom{`(!)}}\ar@{-}[r] &
*++[o][F-]{\verte{`(!)}\phantom{`(!)}}  } \qquad\qquad \xymatrix {
*++[o][F-]{\phantom{`(!)}\phantom{`(!)}}\ar@{-}[r] &
*++[o][F-]{\verte{`(!)}\verte{`(!)}}  }
\] There are two players, \emph{Left} and \emph{Right}.  The leftmost
position above is the winning position for \emph{Left} and the
rightmost position is the winning position for \emph{Right}.  In other
words, the one who owns both token is the winner.  Let us call the
positions $L$, $C,$ and $R$ respectively.  One plays by taking a token
on the opposite node.

\subsubsection{A first tactic: Foresight}
\label{sec:ByL1}

A player realizes that she can win by taking the opponent's token
faster than the opponent can react, i.e., player \emph{Left} can
convert $C$ into $L$ by outpacing player \emph{Right}. Player
\emph{Right}, in turn, can convert $C$ into $R$. This version of the
game has two singleton equilibria: $L$ and $R$.  This is described by
the following conversion
\[
\begin{psmatrix}[colsep=1.2cm,rowsep=1.2cm,nodesep=5pt]
  & [name=L]{L} & [name=C]{C}& [name=R]{R}
  \ncline[arrows=->,linewidth=.03,linecolor=blue,linestyle=dashed]{C}{R}
  \psset{arrowscale=.5}\ncline[arrows=->,linewidth=.08,arrowscale=.6,linecolor=red,linestyle=dashed]{C}{L}
\end{psmatrix}
\]
preference is
\[
\begin{psmatrix}[colsep=1.2cm,rowsep=1.2cm,nodesep=5pt]
  & [name=L]{L} & [name=C]{C}& [name=R]{R}
\ncarc[arrows=->,linewidth=.03,linecolor=blue,linestyle=dotted,arcangle=30]{C}{R}
\ncarc[arrows=->,linewidth=.03,linecolor=blue,linestyle=dotted,arcangle=30]{L}{C}
\psset{arrowscale=.5}
\ncarc[arrows=->,linewidth=.08,arrowscale=.6,linecolor=red,linestyle=dotted,arcangle=30]{C}{L}
\ncarc[arrows=->,linewidth=.08,arrowscale=.6,linecolor=red,linestyle=dotted,arcangle=30]{R}{C}
\end{psmatrix}
\]

\medskip

where $\begin{psmatrix}[colsep=1cm,arrowscale=.5]
  [name=C]& [name=R]
\ncline[arrows=->,linewidth=.08,arrowscale=.6,linecolor=red,linestyle=dotted,arcangle=30]{C}{R}
\end{psmatrix}$ is the preference
for \emph{Left} and $\begin{psmatrix}[colsep=1cm]
  [name=C]& [name=R]
\ncline[arrows=->,linewidth=.03,linecolor=blue,linestyle=dotted,arcangle=30]{C}{R}
\end{psmatrix}$ is the preference for
\emph{Right}.  The change of mind is then:
\[\begin{psmatrix}[colsep=1.2cm,rowsep=1.2cm,nodesep=5pt]
  & [name=L]{L} & [name=C]{C}& [name=R]{R}
 \ncline[arrows=->]{C}{L} 
 \ncline[arrows=->]{C}{R}
\end{psmatrix}
\] and one sees that
there are two equilibria: namely $L$ and $R$, which means that players
have taken both token and keep them.

\subsubsection{A second tactic: Hindsight}
\label{sec:ByL2}

A player, say \emph{Left}, analyzes what would happen if she does not
act. In case \emph{Right} acts, the game would end up in $R$ and
\emph{Left} loses.  As we all know, people hate to lose so they have
an aversion for a losing position. Actually \emph{Left} concludes
that she could have prevented the $R$ outcome by acting. In other
words, it is within \emph{Left}'s power to convert $R$
into~$C$. Similarly for player \emph{Right} from $L$ to~$C$.
\[
\begin{psmatrix}[colsep=1.2cm,rowsep=1.2cm,nodesep=5pt]
  & [name=L]{L} & [name=C]{C}& [name=R]{R}
  \ncline[arrows=->,linewidth=.03,linecolor=blue,linestyle=dashed]{R}{C}
  \psset{arrowscale=.5}\ncline[arrows=->,linewidth=.08,arrowscale=.6,linecolor=red,linestyle=dashed]{L}{C}
\end{psmatrix}
\]

We call naturally
\emph{aversion} the relation that escapes from positions a player does
not want to be, especially a losing position.  Aversion deserves its
name as it works like conversion, but flies from bad position.  We get
the following change of mind:
\[\begin{psmatrix}[colsep=1.2cm,rowsep=1.2cm,nodesep=5pt]
  & [name=L]{L} & [name=C]{C}& [name=R]{R}
 \ncline[arrows=->]{L}{C} 
 \ncline[arrows=->]{R}{C}
\end{psmatrix}
\] 
where $C$ is singleton equilibrium or an Abstract Nash Equilibrium.

\subsubsection{A third tactic: Omnisight}
\label{sec:ByL3}

The players have both hindsight and foresight, resulting in a CP game
\[
\begin{psmatrix}[colsep=1.2cm,rowsep=1.2cm,nodesep=5pt]
  & [name=L]{L} & [name=C]{C}& [name=R]{R}
\ncarc[arrows=->,linewidth=.03,linecolor=blue,linestyle=dashed,arcangle=30]{C}{R}
\ncarc[arrows=->,linewidth=.03,linecolor=blue,linestyle=dashed,arcangle=30]{L}{C}
\psset{arrowscale=.5}
\ncarc[arrows=->,linewidth=.08,arrowscale=.6,linecolor=red,linestyle=dashed,arcangle=30]{C}{L}
\ncarc[arrows=->,linewidth=.08,arrowscale=.6,linecolor=red,linestyle=dashed,arcangle=30]{R}{C}
\end{psmatrix}
\]
with one change-of-mind equilibrium covering all
outcomes thus, no singleton equilibrium (or Abstract Nash Equilibrium) exists.
\[
\begin{psmatrix}[colsep=1.2cm,rowsep=1.2cm,nodesep=5pt]
  & [name=L]{L} & [name=C]{C}& [name=R]{R}
\ncarc[arrows=->,arcangle=30]{C}{R}
\ncarc[arrows=->,arcangle=30]{L}{C}
\ncarc[arrows=->,arcangle=30]{C}{L}
\ncarc[arrows=->,arcangle=30]{R}{C}
\end{psmatrix}
\]
\subsubsection{A four tactic: Defeatism}
\label{sec:defeat}

One of the player, say \emph{Left}, acknowledges that she will be
outperformed by the other, \emph{Right} in this case.  She is so
terrified by her opponent that she returns the token when she has it.
This yields the following conversion:
\[
\begin{psmatrix}[colsep=1.2cm,rowsep=1.2cm,nodesep=5pt]
  & [name=L]{L} & [name=C]{C}& [name=R]{R}
  \ncline[arrows=<-,linewidth=.03,linecolor=blue,linestyle=dashed]{R}{C}
  \psset{arrowscale=.5}\ncline[arrows=->,linewidth=.08,arrowscale=.6,linecolor=red,linestyle=dashed]{L}{C}
\end{psmatrix}
\]
We get the following change of mind:
\[\begin{psmatrix}[colsep=1.2cm,rowsep=1.2cm,nodesep=5pt]
  & [name=L]{L} & [name=C]{C}& [name=R]{R}
 \ncline[arrows=->]{L}{C} 
 \ncline[arrows=->]{C}{R}
\end{psmatrix}
\] 
where $R$ is a singleton equilibrium or an Abstract Nash Equilibrium.  

\subsubsection{Relation with evolutionary games}
\label{sec:evol}

In~\cite{nowak04:_evolut_dynam_of_biolog_games} (Fig. 2, p. 795),
Nowak and Sigmund comment a similar situation in evolutionary games.
They call the first tactic, \emph{bistability}, the second tactic,
\emph{coexistence}, the third tactic, \emph{neutrality} and the fourth
tactic, \emph{dominance} and exhibit the same pictures.

The changes of mind corresponding to the four tactics and their
correspondence with evolutionary games with two strategies can be
summarized as follows

\medskip

\begin{center}
  \begin{tabular}{||l|l|l||}
    \hline\hline
   \textbf{ Blink you lose} & \textbf{Change of Mind} & \textbf{Evol. games}
    \\\hline\hline
    Foresight & 
    \begin{psmatrix}[colsep=1.2cm,rowsep=1.2cm,nodesep=5pt]
      [name=L]{L} & [name=C]{C}& [name=R]{R} \ncline[arrows=->]{C}{L}
      \ncline[arrows=->]{C}{R}
    \end{psmatrix} 
    & Bistablity \\[5pt]\hline 
    Hindsight &
    \begin{psmatrix}[colsep=1.2cm,rowsep=1.2cm,nodesep=5pt]
      [name=L]{L} & [name=C]{C}& [name=R]{R} \ncline[arrows=->]{L}{C}
      \ncline[arrows=->]{R}{C}
    \end{psmatrix} 
    & Coexistence \\[5pt]\hline 
    Omnisight &
    \begin{psmatrix}[colsep=1.2cm,rowsep=1.2cm,nodesep=5pt]
      [name=L]{L} & [name=C]{C}& [name=R]{R}
      \ncarc[arrows=->,arcangle=20]{C}{R}
      \ncarc[arrows=->,arcangle=20]{L}{C}
      \ncarc[arrows=->,arcangle=20]{C}{L}
      \ncarc[arrows=->,arcangle=20]{R}{C}
    \end{psmatrix}
    & Neutrality \\[5pt]\hline 
    Defeatism &
    \begin{psmatrix}[colsep=1.2cm,rowsep=1.2cm,nodesep=5pt]
      [name=L]{L} & [name=C]{C}& [name=R]{R} \ncline[arrows=->]{L}{C}
      \ncline[arrows=->]{C}{R}
    \end{psmatrix}
    & Dominance
    \\[5pt]\hline\hline
  \end{tabular}
\end{center}

\subsection{The $`l$ phage as a CP game}
\label{sec:l_phage}

The $`l$ phage is a game inspired from
biology~\cite{ptashne04:_genet_switc,ptashne01:_genes_and_signal}.
The origin of the game will be given in Section~\ref{sec:grn_as_cpg},
here we give just \emph{the rules of the game.}

There are three players $cI$, $cro$ and $Env$. The game can be seen as
a game with two tokens moving on two graphs where each player may
choose to move one of the two tokens\footnote{In the asynchronous
version.}.  $Env$ moves one token from the bottom position.  The
conversion is therefore the same for the three players\footnote{Note
the difference with the square game where players had different
conversions and the same preference.  The fact that the conversion is
the same for everybody seems to be a feature of biologic
game. Moreover notice also that, unlike in strategic games, the
conversion is not transitive.}  and is given by the following
rightmost diagram:
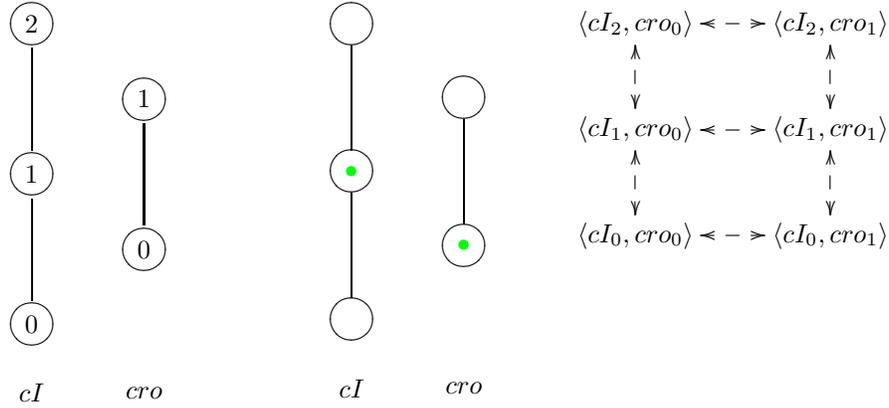
\[ \xymatrix @R 10pt { *++[o][F-]{2}\\ &*++[o][F-]{1}&\\
*++[o][F-]{1}\ar@{-}[uu]\\ &*++[o][F-]{0}\ar@{-}[uu]\\
*++[o][F-]{0}\ar@{-}[uu]\\ cI & cro } \qquad\quad \xymatrix @R 11.5pt
{ *++[o][F-]{\phantom{`(!)}}\\ &*++[o][F-]{\phantom{`(!)}}&\\
*++[o][F-]{\verte{`(!)}}\ar@{-}[uu]\\
&*++[o][F-]{\verte{`(!)}}\ar@{-}[uu]\\
*++[o][F-]{\phantom{`(!)}}\ar@{-}[uu]\\ cI & cro } 
\xymatrix{ \l cI_2,cro_0 \r \ar@{<-->}[r] \ar@{<-->}[d] & \l cI_2,cro_1
\r\ar@{<-->}[d] \\ \l cI_1,cro_0 \r \ar@{<-->}[r] \ar@{<-->}[d] & \l
cI_1,cro_1 \r\ar@{<-->}[d] \\ \l cI_0,cro_0 \r \ar@{<-->}[r] & \l
cI_0,cro_1 \r }
\] The preference is difficult to describe as an actual game to be
played, it comes from the genetics and is specific to each player.
The philosophy is as follows: a gene prefers a position if it is
``pushed forward'' that position.  
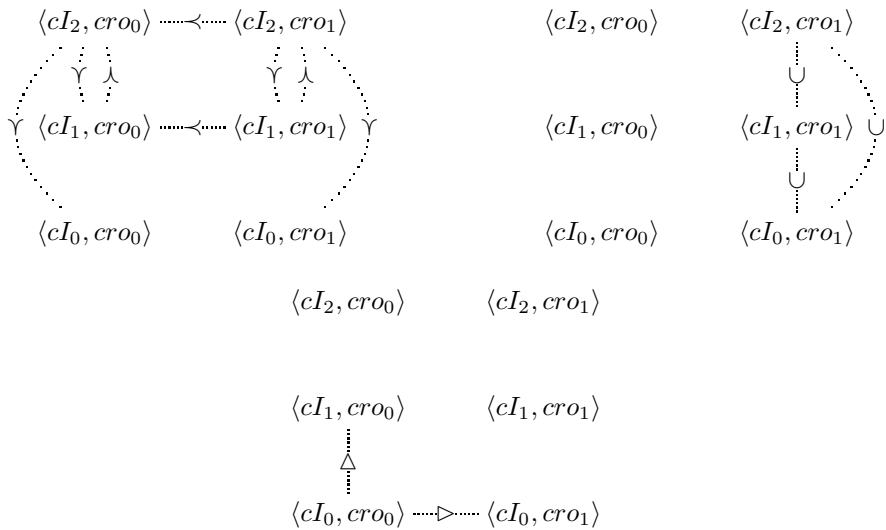
\label{preference} 
  \[ \xymatrix{\l cI_2,cro_0 \r
\ar@{.}@/_/[d]|{\rotatebox{270}{$\succ$}}
\ar@{.}@/_2.5pc/[dd]|{\rotatebox{270}{$\succ$}} & \l cI_2,cro_1 \r
\ar@{.}@/_/[d]|{\rotatebox{270}{$\succ$}}
\ar@{.}@/^2.5pc/[dd]|{\rotatebox{270}{$\succ$}}
\ar@{.}[l]|{\rotatebox{180}{$\succ$}}\\ \l cI_1,cro_0 \r
\ar@{.}@/_/[u]|{\rotatebox{90}{$\succ$}} & \l cI_1,cro_1 \r
\ar@{.}@/_/[u]|{\rotatebox{90}{$\succ$}}\ar@{.}[l]|{\rotatebox{180}{$\succ$}}\\
\l cI_0,cro_0 \r & \l cI_0,cro_1 \r } \qquad \qquad \qquad %
\xymatrix{ \l cI_2,cro_0 \r & \l cI_2,cro_1 \r
\ar@{.}[d]|{\rotatebox{270}{$\supset$}}
\ar@{.}@/^2.5pc/[dd]|{\rotatebox{270}{$\supset$}} \\ \l cI_1,cro_0 \r
& \l cI_1,cro_1 \r \ar@{.}[d]|{\rotatebox{270}{$\supset$}} \\ \l
cI_0,cro_0 \r & \l cI_0,cro_1 \r } %
\]
\[ \xymatrix{ \l cI_2,cro_0 \r & \l cI_2,cro_1 \r \\ \l cI_1,cro_0 \r
& \l cI_1,cro_1 \r \\ \l cI_0,cro_0 \r
\ar@{.}[u]|{\rotatebox{270}{$\triangleleft$}}
\ar@{.}[r]|{\rotatebox{180}{$\triangleleft$}}& \l cI_0,cro_1 \r } %
    \] \label{preference2} From the conversion and the preferences one
deduces three changes of mind.
  
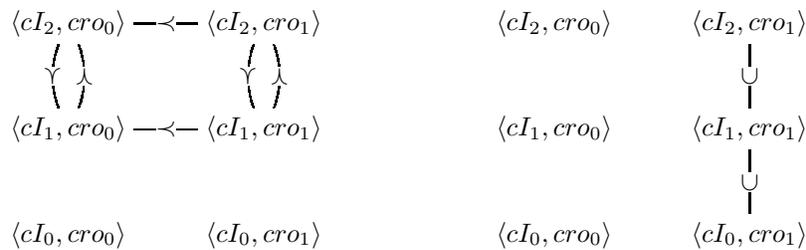
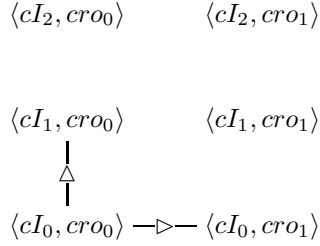
\[ \xymatrix{\l cI_2,cro_0 \r
\ar@{-}@/_/[d]|{\rotatebox{270}{$\succ$}} & \l cI_2,cro_1 \r
\ar@{-}@/_/[d]|{\rotatebox{270}{$\succ$}}
\ar@{-}[l]|{\rotatebox{180}{$\succ$}}\\ \l cI_1,cro_0 \r
\ar@{-}@/_/[u]|{\rotatebox{90}{$\succ$}} & \l cI_1,cro_1 \r
\ar@{-}@/_/[u]|{\rotatebox{90}{$\succ$}}\ar@{-}[l]|{\rotatebox{180}{$\succ$}}\\
\l cI_0,cro_0 \r & \l cI_0,cro_1 \r } \qquad \qquad \qquad %
\xymatrix{ \l cI_2,cro_0 \r & \l cI_2,cro_1 \r
\ar@{-}[d]|{\rotatebox{270}{$\supset$}} \\ \l cI_1,cro_0 \r &
\l cI_1,cro_1 \r \ar@{-}[d]|{\rotatebox{270}{$\supset$}} \\
\l cI_0,cro_0 \r & \l cI_0,cro_1 \r } %
    \]
\[ \xymatrix{ \l cI_2,cro_0 \r & \l cI_2,cro_1 \r \\ \l cI_1,cro_0 \r
& \l cI_1,cro_1 \r \\ \l cI_0,cro_0 \r
\ar@{-}[u]|{\rotatebox{270}{$\triangleleft$}}
\ar@{-}[r]|{\rotatebox{180}{$\triangleleft$}}& \l
cI_0,cro_1 \r } %
    \] from which we deduce the (general) change of mind of the game:
   
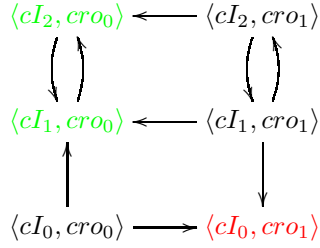
\[ \xymatrix{\verte{\l cI_2,cro_0 \r \ar@{->}@/_/[d]} & \l
cI_2,cro_1 \r \ar@{->}@/_/[d] \ar@{->}[l]\\ \verte{\l cI_1,cro_0
\r \ar@{->}@/_/[u]} & \l cI_1,cro_1 \r
\ar@{->}@/_/[u]\ar@{->}[l]\ar@{->}[d] \\ \l cI_0,cro_0 \r
\ar@{->}[u] \ar@{->}[r] & \rouge{\l cI_0,cro_1 \r} }
    \] One sees one singleton equilibrium namely $\l cI_0,cro_1\r$
(called the \emph{lyse}) and one dynamic equilibrium namely $\{\l
cI_2,cro_0\r, \l cI_1,cro_0\r\}$ (called the \emph{lysogen}).

\section{Formal presentation of CP games}
\label{sec:CP_games}

To define a CP game we have to define four concepts:
\begin{itemize}
\item a set \A{} of \emph{agents},
\item a set \S{} of \emph{situations},
\item for every agent $a$ a relation $\conv{a}$ on \S, called
\emph{conversion},
\item for every agent $a$ a relation $\pref{a}$ on \S, called
\emph{preference}.
\end{itemize}

From these relations we are going to define a relation called
\emph{change of mind}.  Before let us introduce formally what a game
is.

\begin{definition}[Game]
A game is a 4-uple $\l \A, \S,(\conv{a})_{a`:\A}, (\pref{a})_{a`:\A}\r$.
\end{definition}

\begin{example}[Square game 1rst version] For the first version of the
square game we have:
  \begin{itemize}
  \item $\A = \{\Alice, \Beth\}$,
  \item $\S = \{1|2, 1|3, 1|4, 2|3, 2|4, 2|1, 3|4, 3|1, 3|2, 4|1, 4|2,
4|3\}$,
  \item Conversions $\conv{\Alice}$ and $\conv{\Beth}$ are given by
Figure~\ref{fig:conv_pref_square} left,
  \item $\pref{\Alice}$ is the same as $\pref{\Beth}$ and this
relation is given by Figure~\ref{fig:conv_pref_square} right.
  \end{itemize}
\end{example}

\subsection{Abstract Nash equilibrium or singleton equilibrium}
\label{sec:sing_seq}

Let us look at a first kind of equilibria.

\begin{definition}[Abstract Nash equilibrium or singleton equilibrium]
A \emph{singleton equilibrium} is a situation $s$ such that:
\[ `A a`:\A, s'`:\S \quad . \quad (s \conv{a} s') \ \Longrightarrow\
\neg(s \pref{a} s').\] We write $\EqFct{aN}{\cpG{}}{s}$ (aN stands for
\emph{abstract Nash}).
\end{definition}

In the previous paragraphs, we have seen examples of singleton
equilibria.  If we are at such an equilibrium, this is fine, but if
not, we may wonder how to reach an equilibrium.  If $s$ is not an
equilibrium, this means that $s$ fulfills
\[ `E s'`:\S \quad . \quad s \conv{a} s' \wedge s \pref{a} s'\] which
is the negation of
\[`A s'`:\S \quad . \quad (s \conv{a} s') \ \Longrightarrow\ \neg(s
\pref{a} s').\] The relation $s \conv{a} s' \wedge s \pref{a} s'$
between $s$ and $s'$ is a derived one.  Let us call it \emph{change of
mind for $a$} and write it $\CoM{a}$.  We say that $a$ changes her
mind, because she is not happy with $s$ and hopes that following
$\CoM{a}$ she will reach not necessary the equilibrium, but at least a
better situation.  Actually since we want to make everyone happy, we
have to progress along all the $\CoM{a}$'s.  Thus we consider a more
general relation which we call just \emph{change of mind} and which is
the union of the $\CoM{a}$'s. We define this new relation as the union
of the changes of mind of the agents.  \[\gCoM \ \bydef \
\bigcup_{a`:\A} \CoM{a}.\] Now suppose that we progress along $\gCoM$.
What happens if we reach an $s$ from which we cannot progress further?
This means
\[`A a`:\A, s'`:\S \quad . \quad \neg (s \conv{a} s' \wedge s \pref{a}
s')\] in other words, $s$ is an equilibrium.  Hence to reach an
equilibrium, we progress along $\gCoM$ until we are stuck.  In graph
theory, a vertex from which there is no outgoing arrow is called an
\emph{end point} or a \emph{sink}.  In relation theory it is called a
\emph{minimal element}:
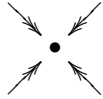
\[\xymatrix @C 12pt @R 12pt {\ar@{->>}[dr] && \ar@{->>}[dl]\\ &`(!)
&\\ \ar@{->>}[ur] && \ar@{->>}[ul]}
\] Thus we look for end points in the graph.
\subsection{Dynamic equilibrium}
\label{sec:dyn_eq}

Actually this progression along $\gCoM$ is not the panacea to reach an
equilibrium. Indeed it could be the case that this progression never
ends, since we enter a perpetual move (think at the square game 2nd
version, Figure~\ref{fig:conv_CoM_square3}).  Actually we identify
this perpetual move as a second kind of equilibrium.

\subsubsection{Strongly connected components}
\label{sec:SCC}

Here it is relevant to give some concepts of graph theory.  A
graph\footnote{In this paper, when we say ``graph'' , we mean always
``oriented graph'' or ``digraph''.} is \emph{strongly connected}, if
given two nodes $n_1$ and $n_2$ there is always a path going from
$n_1$ to $n_2$ and a path going from $n_2$ to $n_1$.  Not all the
graphs are strongly connected, but they may contain some maximal
subgraphs that are strongly connected; ``maximal'' means that one
cannot add nodes without breaking the strong connectedness.  Such a
strongly connected subgraph is called a \emph{strongly connected
components}, \emph{SCC} in short.

\agamewithtwoSCCs The graph below has six SCC's:
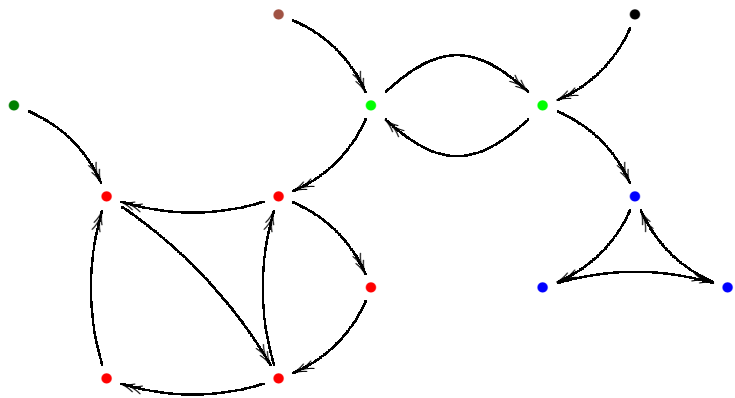
\[ \xymatrix{
&&&\brown{`(!)}\ar@{->>}@/^/[dr]&&&&`(!)\ar@{->>}@/^/[dl]\\
\verdir{`(!)} \ar@{->>}@/^/[dr] &&&&
\verte{`(!)}\ar@{->>}@/^/@(ur,ul)[rr]\ar@{->>}@/^/[dl]&&
\verte{`(!)}\ar@{->>}@/^/@(dl,dr)[ll]\ar@{->>}@/^/[dr]\\ &
\rouge{`(!)}\ar@{->>}@/^/[ddrr] &&\rouge{`(!)}\ar@{->>}@/^/[ll]
\ar@{->>}@/^/[dr]&&&& \bl{`(!)}\ar@{->>}@/^/[dl]&\\
&&&&\rouge{`(!)}\ar@{->>}@/^/[dl]&& \bl{`(!)}\ar@{->>}@/^/[rr]&&
\bl{`(!)}\ar@{->>}@/^/[ul]\\
&\rouge{`(!)}\ar@{->>}@/^/[uu]&&\rouge{`(!)}\ar@{->>}@/^/[ll]\ar@{->>}@/^/[uu]&&&&&
}
\] The graph of Figure~\ref{fig:conv_CoM_square3} has two SCC's
(Figure~\ref{fig:2SCC}).

From a graph, we can deduce a new graph, which we call the
\emph{reduced graph} (or \emph{condensation}~\cite{Baase78}), whose
nodes are the SCC's and the arcs are given as follows: there is an arc
from an SCC $S_1$ to an SCC $S_2$ (assuming that $S_1$ is different
from $S_2$), if there exists a node $n_1$ in $S_1$, a node $s_2$ in
$S_2$ and an arc between $n_1$ and $n_2$.  By construction the reduced
graph has no cycle and its strongly connected components are
singletons.  The reduced graphs associated with the graphs given above
are as follows:
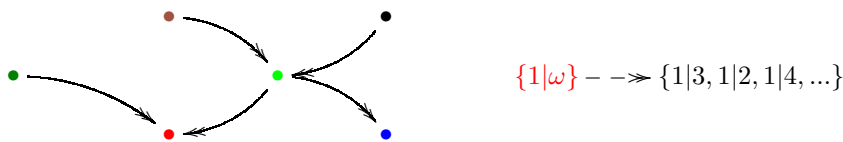
\[ \xymatrix @C 12pt @R 12pt {
&&&\brown{`(!)}\ar@{->>}@/^/[drr]&&&&`(!)\ar@{->>}@/^/[dll]\\
\verdir{`(!)} \ar@{->>}@/^/[drrr]
&&&&&\verte{`(!)}\ar@{->>}@/^/[dll]\ar@{->>}@/^/[drr] \\ &&&
\rouge{`(!)}&&&& \bl{`(!)}  } \qquad\qquad \xymatrix @R 12pt
{&\\\rouge{\{1|`w\}} \ar@{-->>}[r] & \{1|3, 1|2, 1|4,...\}}
\]

\subsubsection{Dynamic equilibria as strongly connected components}
\label{sec:dyn_eq_as_SCC}

At the price of extending the notion of equilibrium, we can prove that
there is always an equilibrium in finite non degenerated games, i.e., in games with a
finite non zero number of game situations.  Indeed given a graph, we compute
its reduced graph.  Then in this reduced graph, we look for end
points.  There is always such an end point since in a finite acyclic graph
(the reduced graph is always acyclic) there exists always at least an
end point.

\begin{center} \doublebox{\parbox{.7\textwidth}{\ \textsl{\emph{CP
Equilibria} are \emph{end points} in the reduced graph.}\ }}
\end{center}

We write $\EqFct{CP}{\cpG{}}{A}$ to say that the subset $A$ of
situations is a CP equilibrium.  We can now split equilibria into two
categories?
\begin{enumerate}
\item \emph{CP Equilibria} (i.e., \emph{Dynamic equilibria}) are
equilibria associated with an SCC and may contain many situations.
\item \emph{Abstract Nash Equilibria} (aka \emph{Singleton
equilibria}) are equilibria associated with an SCC that contains
exactly one situation, i.e., associated with an SCC which is a
singleton, hence the name singleton equilibrium.  There are specific
dynamic equilibria.
\end{enumerate}

Tarjan~\cite{Tarjan} has shown that the reduced graph can be computed
in linear time w.r.t. the numbers of nodes and edges of the original graph.
Therefore CP equilibria can be computed in linear time in the number
of game situations and edges in the change of mind relation, which provides an efficient
algorithm to compute CP equilibria.

\section{What are CP good for?}
\label{sec:motivations}

After the success of strategic games over years, one may wonder why we
introduce a new concept, namely CP games.  The first nice feature is a
theorem that says that \emph{there always exists an equilibrium}.  We
know that pure strategic games do not enjoy that property and that to
obtain such equilibria, Nash had to extend the concept of strategic
game to this of probabilistic games.  Similarly we have relaxed the
notion of equilibrium to this of CP equilibrium.

Beside abstract Nash equilibria that are really like those of strategic games. CP
games have other equilibria that biologists called \emph{dynamic equilibria}
and that correspond to phenomena they actually consider.  Physicists
speak about \emph{stationary states} in that case.

Economists know that the concept of payoff is somewhat
artificial\footnote{See for instance Osborne's introduction of his
  textbook~\cite{osborne04a}.}.  In CP games \emph{the concept of
  payoff is completely abandoned}, no number are attached to
situations and a general relation between situations is proposed instead.

In normal form strategic games, moves from one situation to another
are tightly ruled and strong restrictions are imposed, namely right
and left moves for one player, back and forth moves for another player
and up and down moves for a third, etc., unlike CP games where
\emph{very general moves between situations ruled by the conversions
  are allowed}, like diagonal moves for instance or on the opposite
more restricted moves like horizontal or vertical moves to a neighbor
situation only (see the $\lambda$ phage).  The flexibility of the
conversions and the preferences makes possible to formalize many
situations, like some that occur in biology.

It is known that games are a good framework to analyze models where
the principle of causality fails.  CP~games allow analyzing a larger class of
models.


\section{Gene regulation networks as CP games}
\label{sec:grn_as_cpg}

In the $`l$ phage, levels $0, 1, 2$, for a gene, correspond to levels
of activation or levels of concentration of the corresponding protein.
Thus $cI$ has three levels. $0$ corresponds to the gene being inactive
(the protein is absent), $1$~corresponds to the gene being moderately
active (the protein is present but moderately concentrated),
$2$~corresponds to the gene being highly active (the protein is
concentrated). On the other hand, $cro$ has two levels of activation,
corresponding to the gene being inactive or active.  $Env$ has only
one level, it is always active.  A gene can move from one level at a
time, as translated by the conversion relation on
page~\pageref{sec:l_phage}.  It has been shown that $cI$ is a
repressor for $cro$ and a repressor for itself at level~$2$ and an
activator for itself at level~$1$.  This leads to the preference
$\xymatrix{\ar@{.}[r]|{\succ}&}$ for $cI$ on the left of diagram on
page~\pageref{preference}.  On the other hand, it has been shown that
$cro$ is a repressor for $cI$, this leads to the preference and an
activator for itself at level~$1$.  This leads to the preference
$\xymatrix{\ar@{.}[r]|{\supset}&}$ for $cro$ on the right of diagram on
page~\pageref{preference}.  Moreover when both genes are inactive, the
\emph{environment} may lead to activate either $cI$ or $cro$, this
leads to the preference
$\xymatrix{\ar@{.}[r]|{\triangleright}&}$ of the diagram
on page~\pageref{preference2}.

The two equilibria correspond to two well-known states of the $`l$
phage: the \emph{lyse} and the \emph{lysogen}, which the phage always
reaches.  In particular the lysogen $\{\l cI_2,cro_0\r, \l
cI_1,cro_0\r\}$ is a relatively stable state, where the phage seems
inactive (dormant state). This is due to the fact that the
concentration of the protein associated to $cI$ is controlled: if it
is too concentrated, a repression process makes the concentration to
decrease and vice-versa if the concentration is too low an activation
process makes it to increase.  These antagonistic actions maintain the
concentration at an intermediate level and the associated state is stable.
The state $\l cI_0,cro_1\r$ corresponds to what is called the
\emph{lyse} of the $`l$~phage.

What is amazing in the presentation as CP games is that these states
are actually computed as CP equilibria.  Somewhat connected approaches are~\cite{Thomas73,ThieffryThomas:DynBRN-II-phage-lambda}.

\section{Chinese Wall information security and corporate liability}

A main claim of this article is that CP games, simple as they are, is a natural formalism whose conversion/preference distinction is of wider relevance. We shall further justify the claim in this section, with an example chosen because of the succinctness of its CP-game presentation and because the conversion/preference distinction is of stand-alone interest in the context of the example, without any consideration of equilibrium analysis.\\
The concept of Chinese Wall information security pertains to the prevention of insider trading, and more generally the insulation of insider knowledge \cite{BrewerNash:ChineseWall89}. Chinese Wall requirements are codified in laws in many countries, and are interesting to informaticians in part because Chinese Wall security is different from military-style \emph{need to know}. In particular, any information is in principle accessible to the subjects in question, but access is only granted if the subject is not already in possession of information that could create a conflict of interest. Formally, we consider a set of \emph{subjects}, $P$ (for people), a set of \emph{interests classes}, $I$, a set of \emph{companies}, $C$, with \emph{interest classification} function $\mathcal{I}:C\rightarrow I$, and a set of \emph{objects}, $O$, with \emph{ownership} function $\mathcal{C}:O\rightarrow C$. The typical scenario is that the subjects are consultants and the interest classes consists, for example, of \texttt{bank}, \texttt{oil company}, etc., with the requirement that no consultant handles objects, i.e., information, for more than one, e.g., \texttt{bank}. In other words, we are considering a game played by the subjects, $\A=P$, over complete accounts of what objects each subject has had access to, $\S=\otimes_{p\in P} 2^O$. Writing $s_p$ for the $p$-projection of an $s\in\S$, \cite[Axiom 2]{BrewerNash:ChineseWall89} that governs when a subject, $p$, is allowed to gain access to an object amounts to the conversion relation where $s\conv{p}s'$ iff
\[\begin{array}{l}
    \forall p' \,.\, p\neq p' \Rightarrow s_{p'} = s'_{p'}\\
    \wedge\\
    \exists o \,.\, s'_p = s_p\cup\{o\}
    \wedge
    (\forall o' \in s_p \,.\,
      \mathcal{I}(\mathcal{C}(o)) \neq \mathcal{I}(\mathcal{C}(o'))
      \vee
      \mathcal{C}(o) = \mathcal{C}(o'))
  \end{array}\]
In words, subject $p$ may convert $s$ to $s'$ if the situations only differ by some object, $o$, being added to the $p$-projection and, for all other objects in $s_p$, $o$ either belongs to a different interest class or hails from the same company. By \cite[Axiom 3]{BrewerNash:ChineseWall89}, we are only interested in situations that can be reached from the situation where no subject has had access to any object, $\overrightarrow{\emptyset}$. With this, \cite[Theorem 2]{BrewerNash:ChineseWall89} says: ``A subject can at most have access to one company dataset in each conflict of interest class''. Formally, we have the following.
\begin{definition} A state, $s\in \S$, has \emph{no insider trading} if
\begin{eqnarray*}
\mathrm{NIT}(s) &\triangleq& \forall p\in P \,.\,
 \forall o_1, o_2 \in s_p \,.\,
   \mathcal{I}(\mathcal{C}(o_1)) \neq \mathcal{I}(\mathcal{C}(o_2))
   \vee
   \mathcal{C}(o_1)=\mathcal{C}(o_2)
\end{eqnarray*}
\end{definition}
\begin{theorem} Given a \emph{Chinese Wall CP game form}, $\l P, \otimes_{p\in P} 2^O, (\conv{p})_{p\in P}\r$, no derived \emph{Chinese Wall CP game} can reach a situation with insider trading.
\[\forall(\pref{p})_{p\in P} \,.\, \forall s \,.\,
   \overrightarrow{\emptyset} \CoM{}\!\!^*\ s
   \Rightarrow \mathrm{NIT}(s)\]
\end{theorem}
While logically straightforward, the point of the CP-game version of the theorem is that it is universally quantified over the family of preference relations. In other words, the theorem explicitly states that for a company that implements Chinese Wall regulations for its subjects (the conversion relations), no matter what those employees may be tempted to do (the preference relations), insider trading can only take place if one of the employees breaks the company's rules. This means that the CP-game version of the result formalizes a notion of corporate liability protection, which is directly relevant to the study of Chinese Wall information security.

\section{Conversion or preference, how to choose?}
\label{sec:choose}

The attentive reader may have noticed that what counts to compute
equilibria is the \emph{change of mind} and that keeping the same set
of equilibria there is some freedom on the conversion and the
preference provided one keeps the same change of mind.  More
precisely, we have
  \[
  \begin{array}{rcl@{\qquad}l}
    \CoM{a} & = &  \conv{a} \cap\ \pref{a}\\
    & = & (\conv{a} \cup \ R)\ \cap \pref{a}  & \textrm{when~} R \ \cap
    \pref{a} = \emptyset \\
    & = & \conv{a} \cap\ (\pref{a} \cup T)  & \textrm{when~} T \ \cap
    \conv{a} = \emptyset
  \end{array}
\]

On another hand, one notices that in some examples, the preference is
independent of the agent whereas, in others, the conversion is
independent of the agent.  It seems that this is correlated with the
domain of application.  In particular, we may emit the following
hypothesis.  In biology, conversion is physics and chemistry, whereas
preference is the part that cannot be explained by physics and chemistry, then we may induce that change
of mind (combination of physics and true biology) is life.  Indeed since
physics and chemistry is the same for everyone, it makes sense to say
that conversion is the same for everybody, whereas, due to evolution and biological effects,
preference, changes with agents.

\section{Fixed point construction, and equilibria in infinite games}\label{sect:discrete-FP}

For proving the existence of a fixed point for every probabilistic
game, Nash~\cite{Nash50} used Kakutani's fixed point theorem.  Since we
deal with discrete games, we present in this section a proof of the
existence of equilibrium based on a Tarski fixed-point
theorem~\cite{Tarski:55}. Recall that Tarski's theorem uses an
\emph{update function}, say $f$, on a lattice and builds a fixed point
starting from an element, say $a$, by iteration, $a$, $f(a)$, ...,
$f^n(a)$, ...  Here the lattice is the powerset $\mathscr{P}(\Syn)$
of situations ordered by the subset order.

In analogy to Nash's update function, the function takes a subset of
situations and creates a new subset based on how the agents would like
to improve upon the old subset.

\begin{definition}[Update] Given a game \cpG{} and a subset $\syns\subseteq\Syn$ of the set of situations, let
  $\options{\syns}{\cpG{}}
    \;\eqdef
    \;\bigcup{}_{\syn\in\syns}
      \{\syn'\mid\,\freeComs{}{\syn}{\syn'}\}$.
\end{definition}

With this, we have the following result, covering all CP games.

\begin{lemma}\label{lem:lattice-FP} Given (any) \cpG{}, $\optionsName$ has a complete lattice of fixed points.\footnote{We note that complete lattices are non-empty by definition.}
\end{lemma}

Not all fixed points will correspond to equilibria but the equilibria
are the least, non-empty fixed points of the update function.

\begin{lemma}\label{lem:least-FP-EqN} Given
  \cpG{}, with change-of-mind relation $\freeComName{}$.
\[\begin{array}{c}
     \EqFct{CP}{\cpG{}}{\syns}\\
\Updownarrow\\
    \options{\syns}{\cpG{}}=\syns
    \ \wedge\
    \syns\neq\emptyset
    \ \wedge\
    (\forall \syns' \,.\,
     \emptyset\subsetneq\syns'\subsetneq\syns
       \Longrightarrow
       \options{\syns'}{\cpG{}}\not\subseteq\syns')
  \end{array}\]
\end{lemma}
\begin{proof} By two direct arguments. The only
interesting step is from bottom to top and showing that, for any two
${\syn_1,\syn_2\in\syns}$, we have ${\freeComs{}{\syn_1}{\syn_2}}$.
We first note that $\optionsName$ is post-fixpointed:
${\syns\subseteq\options{\syns}{\cpG{}}}$, idempotent:
${\options{\options{\syns}{\cpG{}}}{\cpG{}}=\options{\syns}{\cpG{}}}$, and order-preserving:
${\syns_1\subseteq\syns_2\,\Longrightarrow\,\options{\syns_1}{\cpG{}}\subseteq\options{\syns_2}{\cpG{}}}$.
By order-preservation and ${\options{\syns}{\cpG{}}=\syns}$, we have
${\options{\{\syn_1\}}{\cpG{}}}\subseteq\syns$. If
${\neg(\freeComs{}{\syn_1}{\syn_2})}$, then
${\syn_2\in\syns\setminus\options{\{\syn_1\}}{\cpG{}}}$, i.e.,
${\options{\{\syn_1\}}{\cpG{}}\subsetneq\syns}$.  By post-fixpointed-ness,
${\options{\{\syn_1\}}{\cpG{}}}$ is non-empty and, by assumption of
least-ness, we may therefore conclude
${\options{\options{\{\syn_1\}}{\cpG{}}}{\cpG{}}\not\subseteq\options{\{\syn_1\}}{\cpG{}}}$.
This contradicts idempotency, and thus
${\freeComs{}{\syn_1}{\syn_2}}$.
\end{proof}

CP equilibria are therefore \emph{atomic}, in the sense that neither
anything smaller nor anything bigger will have the same defining
properties. For finite~\cpG{}, a counting argument shows that the
complete lattice of~$\optionsName$-fixed point will have least,
non-empty elements, thus guaranteeing existence. For the infinite
case, e.g., the following unbounded change-of-mind relation will not
lead to the existence of least, non-empty elements in the fixed-point
lattice because all tails are fixed points.

\[
\xymatrix{
\bullet \ar@{->>}@/^.8pc/[r]&~\bullet \ar@{->>}@/^.8pc/[r]&~\bullet \ar@{->>}@/^.8pc/[r]&~~\bullet\ldots
}
\]


However there are infinite cases where the existence of a CP equilibrium can be guaranteed, namely when there exists an SCC $\syns$, which is extremal for the reduced change of mind.

More generally, a sufficient condition for the existence of $\EqFctName{CP}{\cpG{}}$ is that some $s$ can reach only finitely many other elements in $\S$ using $\freeComsName{}$. The condition is also necessary if we restrict attention to finite $\EqFctName{CP}{\cpG{}}$. In particular, finite games have $\EqFctName{CP}{\cpG{}}$. Because of the role played by reduced graphs above, also games with finite reduced graphs have $\EqFctName{CP}{\cpG{}}$.

\section{Conclusion}
\label{sec:conclusion}

We have presented conversion preference games as a strict extension of
strategic games and we have proved that in a finite CP game an equilibrium always exists and under some conditions in infinite games as well.  This theory is infancy and we expect it to generate as many theorems as the classical Nash game theory.  See for instance~\cite{LeRouxPhD08}.


\begin{thebibliography}{10}

\bibitem{Baase78}
Sara Baase.
\newblock {\em Computer Algorithms: Introduction to Design and Analysis}.
\newblock Addison-Wesley Longman Publishing Co., Inc., Boston, MA, USA, 1978.

\bibitem{BrewerNash:ChineseWall89}
D.~F.~C. Brewer and M.~J. Nash.
\newblock The {C}hinese {W}all security policy.
\newblock In {\em IEEE Symposium on Security and Privacy}, pages 206--214,
  1989.

\bibitem{Nash50}
John~F. Nash{, Jr.}
\newblock Equilibrium points in n-person games.
\newblock {\em Proceedings of the National Academy of Sciences}, 36:48--49,
  1950.

\bibitem{nowak04:_evolut_dynam_of_biolog_games}
M.~A. Nowak and K.~Sigmund.
\newblock Evolutionary dynamics of biological games.
\newblock {\em Science}, 3003:793--798, February 2004.

\bibitem{osborne04a}
Martin~J. Osborne.
\newblock {\em An Introduction to Game Theory}.
\newblock Oxford, 2004.

\bibitem{ptashne04:_genet_switc}
Mark Ptashne.
\newblock {\em Genetic Switch: Phage Lambda Revisited}.
\newblock Cold Spring Harbor Laboratory Press, April 2004.
\newblock 3rd edition.

\bibitem{ptashne01:_genes_and_signal}
Mark Ptashne and Alexander Gann.
\newblock {\em Genes and Signals}.
\newblock Cold Spring Harbor Laboratory Press, November 2001.

\bibitem{LeRouxPhD08}
St\'ephane~Le Roux.
\newblock {\em Abstraction and Formalization in Game Theory}.
\newblock PhD thesis, \'Ecole normale sup\'erieure de Lyon (France), January
  2008.

\bibitem{LeRouxLescanneVestergaard:RGT-Nash}
St{\'e}phane~Le Roux, Pierre Lescanne, and Ren{\'e} Vestergaard.
\newblock A discrete {N}ash theorem with quadratic complexity and dynamic
  equilibria.
\newblock Research Report IS-RR-2006-006, JAIST, May 2006.

\bibitem{Tarjan}
Robert~E. Tarjan.
\newblock Depth first search and linear graph algorithms.
\newblock {\em SIAM Journal on computing}, pages 146--160, Januar 1972.

\bibitem{Tarski:55}
Alfred Tarski.
\newblock A lattice-theoretical fixpoint theorem and its applications.
\newblock {\em Pacific Journal of Mathematics}, 5:285--309, 1955.

\bibitem{ThieffryThomas:DynBRN-II-phage-lambda}
Denis Thieffry and Ren{\'e} Thomas.
\newblock Dynamical behaviour of biological regulatory networks {II}: Immunity
  control in bacteriophage lambda.
\newblock {\em Bulletin of Mathematical Biology}, 57:277--297, 1995.

\bibitem{Thomas73}
Ren{\'e} Thomas.
\newblock Boolean formalization of genetic control circuits.
\newblock {\em Journal of Theoretical Biology}, 42(3):563--585, 1973.

\end{thebibliography}

\end{document}